\documentclass[sn-mathphys-ay]{sn-jnl}

\usepackage[dvipsnames]{xcolor}
\usepackage{graphicx}%
\usepackage{multirow}%
\usepackage{amsmath,amssymb,amsfonts}%
\usepackage{amsthm}%
\usepackage{mathrsfs}%
\usepackage[title]{appendix}%
\usepackage{xcolor}%
\usepackage{textcomp}%
\usepackage{manyfoot}%
\usepackage{booktabs}%
\usepackage{algorithm}%
\usepackage{algorithmicx}%
\usepackage{algpseudocode}%
\usepackage{listings}%
\usepackage{subcaption}

\theoremstyle{thmstyleone}%
\newtheorem{theorem}{Theorem}
\newtheorem{proposition}{Proposition}

\theoremstyle{thmstyletwo}%
\newtheorem{example}{Example}%
\newtheorem{remark}{Remark}%

\theoremstyle{thmstylethree}%
\newtheorem{definition}{Definition}%
\newtheorem{problem}{Problem}%

\newcommand{\lang}{\mathcal{L}}

\newcommand{\mask}{\mathcal{M}}

\usepackage{caption}
\usepackage{subcaption}

\begin{document}

\title[]{Enhancing sensor attack detection in supervisory control systems modeled by probabilistic automata}



\author*[1]{\fnm{Parastou} \sur{Fahim}}\email{pbf5107@psu.edu}

\author[2,3]{\fnm{Samuel} \sur{Oliveira}}\email{samuel.oliveira@unifap.br}

\author[1]{\fnm{R\^omulo} \sur{Meira-G\'oes}}\email{romulo@psu.edu}


\affil[1]{\orgdiv{School of Electrical Engineering and Computer Science}, \orgname{The Pennsylvania State University}, \orgaddress{\city{University Park}, \state{PA}, \country{USA}}}

\affil[2]{\orgdiv{Graduate Program in Electrical Engineering}, \orgname{Santa Catarina State University (UDESC)}, \orgaddress{\city{Joinville}, \state{SC}, \country{Brazil}}}

\affil[3]{\orgdiv{Department of Exact and Technological Sciences}, \orgname{Federal University of Amapá (UNIFAP)}, \orgaddress{\city{Macapá}, \state{AP}, \country{Brazil}}}



\abstract{
Sensor attacks compromise the reliability of cyber-physical systems (CPSs) by altering sensor outputs with the objective of leading the system to unsafe system states. 
This paper studies a probabilistic intrusion detection framework based on $\lambda$-sensor-attack detectability ($\lambda$-sa), a formal measure that evaluates the likelihood of a system being under attack based on observed behaviors.
Our framework enhances detection by extending its capabilities to identify multiple sensor attack strategies using probabilistic information, which enables the detection of sensor attacks that were undetected by current detection methodologies. 
We develop a polynomial-time algorithm that verifies $\lambda$-sa detectability by constructing a weighted verifier automaton and solving the shortest path problem. 
Additionally, we propose a method to determine the maximum detection confidence level ($\lambda$*) achievable by the system, ensuring the highest probability of identifying attack-induced behaviors.  
}

\keywords{Discrete Event Systems; Supervisory Control; Cybersecurity; Intrusion Detection}



\maketitle
\graphicspath{{Figs/}}

\section{Introduction}
\vspace{-0.5em}
Cyber-physical systems (CPSs) integrate physical and computation processes using communication networks to enable the monitoring and control of these processes \citep{allgower2019}. 
In recent years, CPSs have seen an increase in their complexity and dependence on communication networks to ensure safe operation.
However, these new advancements led to new vulnerabilities that cyber attacks can exploit in CPS, e.g., \citep{Farwell:2011,Checkoway:2011,Greenberg:2020,Easterly:2023}.
As key components of CPSs,  sensors and actuators are particularly susceptible to attacks that can compromise the integrity and safety of these systems. 
Therefore, detecting and preventing attacks on sensors and actuators is crucial to ensure the secure and safe operation of CPS.



In this paper, we analyze the security of CPSs at the supervisory control layer in the hierarchical control architecture.
Hence, we use the discrete event modeling formalism of probabilistic discrete event systems (PDES), where system operation and communications are event-based with known transition probabilities and the controller is a supervisor \citep{Lawford:1993,Garg:1999,Lafortune:2021}. 
In this context, the supervisor controls the CPS via actuator commands based on the observation of events generated by sensor readings. 
Based on event-driven models, we assume that an attacker infiltrates and manipulates the sensor communication channels between the plant and the supervisor; this type of attack is known as a sensor deception attack\footnote{Sensor attacks for short.}.
We study the design of \emph{Probabilistic Intrusion Detection (ID) Systems} to detect sensor deception attacks in CPSs. 

In the domain of discrete event systems (DES), several works focused on dealing with cyber attacks, e.g., see \citep{Rashidinejad:2019,hadjicostis2022cybersecurity, OLIVEIRA2023100907}. Among these efforts, ID systems have received significant attention in recent years, particularly for their role in identifying and mitigating such attacks \citep{Thorsley:2006,Carvalho:2018,Lima:2019,meira-goes:2020towards,wang:2022,fritz2023detection,zhang2023robust,lin2024diagnosability, li2025diagnosability, kang2025diagnosability}. 
An ID system monitors the presence of attacks by analyzing the behavior of the controlled system. 
However, the state-of-the-art ID systems in DES has mainly focused in analyzing the controlled behavior qualitatively using logical DES models, i.e., models without probabilities or other quantitative metric.
For this reason, the current ID systems in DES are ineffective against ``smarter attacks" such as stealthy/covert sensor deception attacks \citep{Su:2018,meira-goes:2020synthesis,tong:2022,YAO2024deception}. 

Since ID systems play a crucial role in CPS, there is a need to develop quantitative frameworks for the detection of sensor attacks to complement the logical approach. 
This leads to an important question: \\
\emph{How can we find a quantitative measure to better detect sensor attacks, indicating the certainty that a given behavior originates from the system under attack?}

Our previous works addressed this question by adopting a stochastic framework and proposing a probabilistic ID method based on a property termed $\epsilon$-safety \citep{meira-goes:2020towards, Fahim2024-wodes}. 
This property captures the system's ability to analyze and identify whether an \emph{deterministic} sensor attack strategy \emph{drastically} modifies the probability of the controlled system. 
In other words, $\epsilon$-safety determines if the attacker leaves a \emph{probabilistic footprint} while attacking the controlled system. 
The parameter $\epsilon$ represents the confidence level that an observed behavior is more likely to have been generated by the system under attack than by the nominal system (a system operating under normal conditions). 

In this paper, we extend this problem by exploring a generalization on $\epsilon$-safety. 
Rather than focusing on a specific deterministic sensor attack strategy, we address a wider scope of attack strategies, including \emph{all complete, consistent, and successful} sensor attack strategies.  
While the $\epsilon$-safety framework was limited to detecting a \emph{single sensor attack strategy}, our enhanced approach introduces the notion of $\lambda$-sensor-attack detectability ($\lambda$-sa, for short), which expands the detection capability to handle \emph{multiple} sensor attack strategies. 
This generalization ensures a more robust and adaptable defense mechanism against a wider range of potential threats.

Based on the notion of $\lambda$-sensor-attack detectability, we formulate two novel problems: (i) verifying $\lambda$-sa detectability, and (ii) searching for the largest $\lambda$ such that $\lambda$-sa detectability holds. 
Moreover, we present a \emph{polynomial-time} complexity algorithm to solve these two problems. 
The solution methodology is based on methods from DES and graph theory and consists of two steps. 
First, we construct a structure called the \emph{weighted verifier}, which includes information about both the nominal system and the system under attack. 
This structure is inspired by the verifier automaton used for fault diagnosis purposes, which combines information from faulty and non-faulty systems \citep{yoo2002polynomial}.
By using the weights in the verifier automaton and solving the \emph{shortest path problem} \citep{cormen2022introduction}, we find the string with the least $\lambda$ value.
Specifically, this is achieved by finding the shortest path from the initial state to the marked states within the weighted verifier. 
Solving the shortest path problem over the verifier provides the correct value to check the $\lambda$-sa detectability.

The contributions of this paper are as follows: 
\begin{enumerate}
\item[(1)] The novel definition of $\lambda$-sensor-attack detectability that can ensure probabilistic detection for all complete, consistent, and successful sensor attack strategies. 
\item[(2)] Two new problems formulation based on the $\lambda$-sa detectability. The verification problem of $\lambda$-sa detectability and an optimal value problem to ensure $\lambda$-sa detectability.
\item[(3)] A polynomial-time solution methodology that solves both problems.
\end{enumerate}
The remainder of the paper is organized as follows. Section~\ref{sect:motivation} introduces a motivating example for the problem addressed in this study. 
Section~\ref{sect:sup} reviews the necessary definitions of supervisory control and PDES under sensor deception attacks. 
Section~\ref{sect:preliminaries} presents the modeling of the system under sensor attacks and also the class of all complete, consistent, and successful sensor attack strategies. 
The concept of $\lambda$-sa detectability and two verification problems are formulated in Section~\ref{sect:problem}. 
In Section~\ref{sect:solution}, we describe our polynomial-time complexity solution for the two verification problems. We conclude the paper in Section~\ref{sect:conclusion}. 

\section{Motivation Example} \label{sect:motivation}
Inspired by the problem described in \citep{meira-goes:2020towards}, we consider, as a motivating example, a scenario in which two vehicles are traveling on an infinite, discretized road. In this scenario, the vehicle in front, referred to as \emph{adv veh}, is manually driven, i.e., its actions are uncontrollable. 
The vehicle behind, referred to as \emph{ego veh}, is autonomous.
The two vehicles move in the same direction with an initial distance of 2 units between them.
When the relative distance between the vehicles becomes zero, a collision occurs.

\begin{figure}[h]
    \centering
    \includegraphics[width=0.65\textwidth]{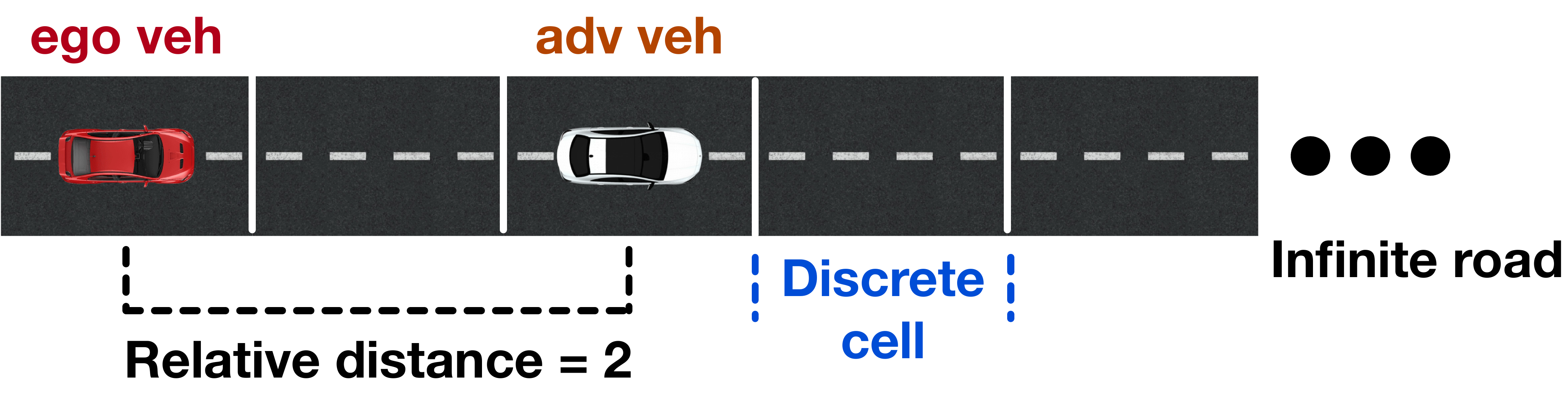} 
    \caption{Overview of discrete collision avoidance modeling for autonomous vehicles} 
    \label{Example} 
\end{figure}

\noindent \textbf{Controlled system.} 
The goal of the $ego$ controller is to avoid collision with $adv$ by measuring the relative distance between the two cars.
Based on a discrete-state event-driven model of this system, we can use standard supervisory control theory techniques to synthesize a controller that ensures no collision.
Intuitively, this safe controller prevents $ego$ from moving ahead when the relative distance between the cars is equal to one. 

\noindent \textbf{System under attack.} 
Let us consider that $ego$ has been compromised by a sensor attacker.
The attacker hijacks the sensors in $ego$ aiming to cause a collision between $ego$ and $adv$. 
The attacker might use two different attack strategies to cause the collision between the cars.
The first attack strategy, $att_1$, will immediately insert a fictitious reading of relative distance equal to $2$ when $ego$ is only $1$ cell away from $adv$. 
In this manner, the controller allows $ego$ to move forward and collide with $adv$. 
On the other hand, the second attack strategy, $att_2$, only makes the insertion the second time $ego$ and $adv$ are $1$ cell away.
The first attack strategy ``eagerly" changes the nominal behavior to reach a collision whereas the second attack strategy allows the nominal behavior to happen before changing it.

\noindent \textbf{ID systems.}
Using probabilistic information about the system allows detection analysis beyond the logical ID systems.
We discuss ID mechanisms options below. 

\noindent \textbf{(1) Logical ID systems.} 
Logical ID mechanisms rely on monitoring the behavior of the controlled system to determine whether an attack has occurred or not. 
These mechanisms identify a sensor attack when the observed behavior deviates from the nominal behavior \citep{Carvalho:2018, Lima:2019, lin2024diagnosability}. 
Since a relative distance of $2$ is possible in the collision avoidance system, both $att_1$ and $att_2$ strategies feed the ID with nominal behavior.
Thus, logical ID systems cannot detect strategies $att_1$ and $att_2$.

\noindent \textbf{(2) Probabilistic $\epsilon$-safety ID system.} 
Our previous works \citep{meira-goes:2020towards, Fahim2024-wodes} provided a probabilistic approach to detect a sensor attack strategy using the notion of $\epsilon$-safety. 
This notion compares the probability of generating a nominal behavior versus an attacked one.
According to \citep{Fahim2024-wodes}, this system is $\epsilon$-safe with respect to strategy $att_1$ for a confidence level of $0.9$.
The ID detects $att_1$ since the attacker eagerly inserts the fictitious event without considering its probability changes.
However, $\epsilon$-safety regarding strategy $att_1$ does not provide any information about detecting strategy $att_2$.
We need to verify $att_2$ to guarantee that $\epsilon$-safety holds, i.e., a new run of the verification procedure.

\noindent \textbf{(2) Probabilistic $\lambda$-sensor-attacks ID system.} 
We present an enhanced probabilistic ID approach that goes beyond detecting a specific sensor attack strategy. 
The notion of $\lambda$-sensor-attacks detectability can effectively detect all possible sensor attack strategies.
In Section~\ref{sect:solution}, we show that the collision avoidance system in Fig.~\ref{Example} is $\lambda$-sensor-attacks detectable with confidence level $0.9$.
It means that both $att_1$ and $att_2$ can be detected when using probabilistic information of the system.

\section{Modeling of Controlled Systems}\label{sect:sup} 
\subsection{Supervisory Control}\label{subsect:des}
We consider the supervisory layer of a feedback control system, where the uncontrolled system (plant) is modeled as a Deterministic Finite-State Automaton (DFA) in the discrete-event modeling formalism.
A DFA is defined by $G := (X_G,\Sigma,\delta_G,x_{0,G},X_{m,G})$, where $X_G$ is the finite set of states, $\Sigma$ is the finite set of events, $\delta_G:X_G\times\Sigma\rightarrow X_G$ is the partial transition function, $x_{0,G}$ is the initial state, and $X_{m,G}$ is the set of marked states.
The function $\delta_G$ is extended, in the usual manner, to the domain $X_G\times\Sigma^*$. 
The language and the marked language generated by $G$ are defined by $\lang(G) := \{s \in \Sigma^*\mid \delta_G(x_{0,G},s)!\}$ and $\lang_m(G) := \{s \in \lang(G)\mid \delta_G(x_{0,G},s)\in X_{m,G}\}$, where $!$ means that the function is defined.

For convenience, we define $\Gamma_G(x) := \{e\in\Sigma\mid\delta_G(x,e)!\}$ as the set of feasible events in state $x\in X_G$.
For string $s$, the length of $s$ is denoted by $|s|$ whereas $s[i]$ denotes the $i^{th}$ event of $s$ such that $s = s[1]s[2]\ldots s[|s|]$.
The $i^{th}$ prefix of $s$ is defined by $s^i$, i.e., $s^i = s[1]s[2]\ldots s[i]$ and $s^0 = \epsilon$. 

Considering the supervisory control theory of DES \citep{Ramadge:1987}, a \emph{supervisor} controls the plant $G$ by dynamically disabling events.
The limited actuation capability of the supervisor is modeled by partitioning the event set $\Sigma$ into the sets of controllable and uncontrollable events, $\Sigma_{c}$ and $\Sigma_{uc}$.
The supervisor cannot disable uncontrollable events. 
Therefore, the supervisor's control decisions are limited to the set $\Gamma:=\{\gamma\subseteq\Sigma\mid\Sigma_{uc} \subseteq \gamma\}$.
Formally, a supervisor is a mapping $S:\lang(G)\rightarrow\Gamma$ defined to satisfy specifications on $G$, e.g., avoid unsafe states in $G$, avoid deadlocks.

The closed-loop behavior of $G$ under the supervision of $S$ is denoted by $S/G$ and generates the closed-loop languages $\lang(S/G)$ and $\lang_m(S/G)$; see, e.g., \citep{Lafortune:2021,Wonham:2018}.
Herein, we assume that supervisor $S$ is realized by an automaton $R = (X_R,\Sigma,\delta_R,x_{0,R})$, i.e., $S(s) = \Gamma_R(\delta_R(x_{0,R},s))$.
With an abuse of notation, we use $S$ and $R$ interchangeably hereafter.

\subsection{Stochastic Supervisory Control}\label{subsect:sdes}
We consider a stochastic DES modeled as a probabilistic discrete event system (PDES) defined similarly to a DFA \citep{Lawford:1993,Garg:1999,Pantelic:2014}. 
A PDES is defined by the tuple $G := (X_G,\Sigma,\allowbreak \delta_G, P_G,x_{0,G},X_{m,G})$ where $X_G$, $\Sigma$, $\delta_G$, $x_{0,G}$ and $X_{m,G}$ are defined as in the DFA definition.
The probabilistic transition function $P_G:X_G\times\Sigma\times X_G\rightarrow [0,1]$ specifies the probability of moving from state $x$ to state $y$ with event $e$. 
Hereupon, we will use the notation $G$ to describe a PDES.

In this work, we assume that $G$ is deterministic where $\nexists y,y^*\in X_G$, $y^*\neq y$ such that $P_G(x,e,y)>0$ and $P_G(x,e,y^*)>0$.
In this manner, $\delta_G(x,e)= y$ if and only if $P_G(x,e,y) > 0$.
Moreover, we assume that each state in $G$ transitions with probability $1$ or deadlocks, i.e.,  $\sum_{e\in \Sigma} \sum_{y\in X_H} P_G(x,e,y) \in \{0,1\}$ for any $x \in X_G$.

\begin{example}  
Figure~\ref{fig:plant_G} provides an example of PDES $G$ modeling the collision avoidance motivating example described in Sect.~\ref{sect:motivation}.
The PDES $G$ has four states and three events.
States in $G$ model the relative distance among $ego$ and $adv$ whereas events $a,\ b$, and $c$ model the decrease, increase, and no change, respectively, on this relative distance.
The cars collide when the relative distance is equal to zero.
For simplicity, once the relative distance is greater than or equal to three, $adv$ escapes $ego$'s range.
States $0$ and $3$ are deadlock states, e.g., $\sum_{e\in \Sigma} P_G(0,e,\delta_G(x,e)) = 0$.
On the other hand, states $1$ and $2$ are live states with probability $1$.
The dynamics of $G$ is captured by the transitions in Fig.~\ref{fig:plant_G} where the label describes the event and probability of the transition, respectively.
For instance, transition $2\rightarrow 1$ has probability $0.1$ of occurring, e.g, $P_G(2,a,1) = 0.1$.
\end{example}

\begin{figure}[thpb]
\centering
\includegraphics[width=0.4\columnwidth]{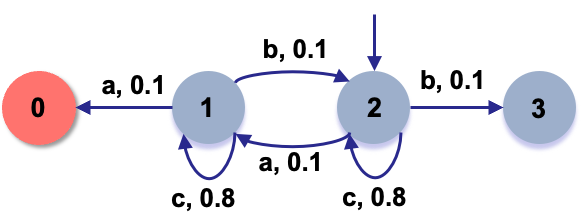}
\caption{PDES $G$ collision avoidance}
\label{fig:plant_G}
\vspace{-2em}
\end{figure}

Although the language and marked language of PDES $G$ are defined as in the DFA case, the notion of probabilistic languages (p-languages) of a PDES was introduced to characterize the probability of executing a string \citep{Garg:1999}.
The p-language of $G$, $L_p(G):\Sigma^*\rightarrow [0,1]$, is recursively defined for $s\in\Sigma^*$ and $e\in\Sigma$ as: 
$L_p(G)(\epsilon) := 1$, $L_p(G)(se) := L_p(G)(s)P_G(x,e,y)$ if $x=\delta_G(x_{0,G},s)$,  $e\in \Gamma_G(x)$ and $y = \delta_G(x,e)$, and $0$ otherwise.
\begin{example}  
Continuing with our running example, we can obtain the probabilistic language of PDES $G$ in Fig.~\ref{fig:plant_G}.
For instance, the probability of executing string $abc \in \lang(G)$ is inductively computed by $P_G(2,a,1)P_G(1,b,2)P_G(2,c,2) = 0.1\times 0.1\times 0.8$.
\end{example}

For convenience, we write $P_G(x,e)$ to denote $P_G(x,e,\delta_G(x,e))$ whenever $\delta_G(x,e)!$, i.e., the probability of executing $e$ in state $x$.

In the context of stochastic supervisory control theory, the plant $G$ is controlled by a supervisor $S$ as described in Section~\ref{subsect:des}.
In this work, we use the framework of supervisory control of PDES introduced by \citep{Kumar:2001}.
The supervisor is \emph{deterministic} and realized by DFA $R = (X_R,\Sigma,\allowbreak \delta_R,x_{0,R})$ as previously-described.
And although $R$ is deterministic, its events disablement proportionally increases the probability of the enabled ones.
Given a state $x\in X_G$, a state $y \in X_R$, and an event $e \in \Gamma_G(x)\cap\Gamma_R(y)$, the probability of $e$ being executed is given by the standard normalization:
\begin{equation}\label{eq:renormalization}
P_{R,G}((x_R,x_G),e,(y_R,y_G)) = \frac{P_G(x_G,e)}{\sum_{\sigma\in\Gamma_G(x_R)\cap\Gamma_R(y_R)}P_G(x_G,\sigma)}
\end{equation}
Due to this renormalization, the closed-loop behavior $R/G$ generates a p-language different, in general, than the p-language of $G$.
We can represent the closed-loop behavior $R/G$ as PDES $R||_p G$ where $||_p$ is defined based on the parallel composition $||$ and Eq.~\ref{eq:renormalization}.
The formal definition of $||_p$ is described in Appendix~\ref{app:parallel_prob}.

\begin{example}
Regarding our running example, the supervisor $R$ for PDES $G$ is shown in Fig.~\ref{fig:sup_R}.
This supervisor ensures that the plant never reaches state $0$.
The closed-loop representation of $R/G$ is shown in Fig.~\ref{fig:M_n}.
The supervisor disables event $a$ in state $1$ which results in disabling event $a$ in state $1$ in the plant.
For this reason, events from state $(1,1)$ in $R/G$ have been normalized based on Eq.~\ref{eq:renormalization}.
For example, $P_{R,G}((1,1),b,(2,2)) = \frac{0.1}{0.9}= 0.111\dots$.
\end{example}
\begin{figure}[thpb]
\begin{subfigure}[t]{0.45\columnwidth}
\centering
\includegraphics[width=0.75\columnwidth]{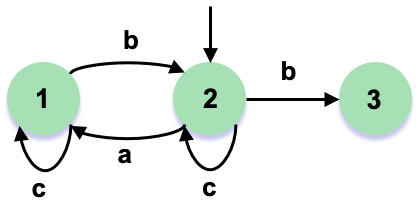}
\caption{Supervisor $R$ collision avoidance}
\label{fig:sup_R}
\end{subfigure}
\ 
\begin{subfigure}[t]{0.45\columnwidth}
\centering
\includegraphics[width=0.75\columnwidth]{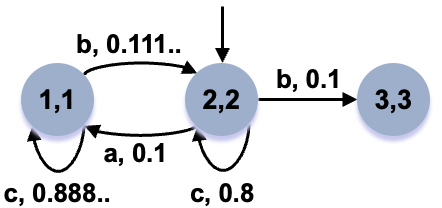}
\caption{Supervised system $R/G=R||_pG$}
\label{fig:M_n}
\end{subfigure}
\caption{Supervisor $R$ and controlled system $R/G$}
\label{fig:supervisor-controlled}
\vspace{-2em}
\end{figure}

For simplicity, we assume that the plant $G$ has one critical state, denoted $x_{crit} \in X_G$.
Moreover, we assume that every transition to $x_{crit}$ is controllable, i.e., $\delta_G(x,e) = x_{crit} \Rightarrow e\in \Sigma_c$ for any $x\in X_G$.
These assumptions are without loss of generality as we can generalize it to any regular language by state space refinement in the usual way \citep{Cho:1989,Lafortune:2021}. 
Supervisor $R$ ensures the critical state is not reachable in $R/G$ as in our running example.

\section{Modeling of Control Systems under Sensor Attacks}\label{sect:preliminaries}
In a cyber-security context, we assume that the supervisory control system can be under sensor attack.
In sensor attacks, an attacker hijacks and controls a subset of the sensors to reach the critical state.
In this manner, this attacker modifies the closed-loop behavior of the system.
In this section, we review the supervisory control under sensor deception attacks as in \citep{Su:2018,meira-goes:2020synthesis,meira-goes:2023dealing}.
We focus on explaining the probabilistic attacker model and defining the control system under attack.

\subsection{Sensor Attacker} 
We follow the probabilistic sensor attacker model introduced in \citep{meira-goes:2021synthesis}.
A sensor deception attacker compromises a subset of sensor events, denoted $\Sigma_a\subseteq \Sigma$.
This attacker can modify the readings of these compromised events by inserting fictitious events into or deleting event readings from the supervisor.
To identify the attacker actions, insertion and deletion events are modeled using the sets $\Sigma_i = \{ins(e) \mid e \in \Sigma_a\}$ and $\Sigma_d=\{del(e)\mid e \in \Sigma_a\}$, respectively.
The union of the insertion, deletion, and plant event sets, $\Sigma_m = \Sigma\cup\Sigma_i\cup\Sigma_d$, encompasses the event set of the system under attack. 
Formally, the attacker is defined as:

\begin{definition}[Attack strategy]\label{def:attack_str}
An attack strategy with compromised event set $\Sigma_a$ is defined as a partial map $A: \Sigma_m^* \times (\Sigma\cup \{\epsilon\})\rightarrow \Sigma_m^*$ that satisfies for any $t \in \Sigma^*_m$ and $e \in \Sigma\cup \{\epsilon\}$:
\begin{enumerate}
    \item $A(\epsilon,\epsilon) \in \Sigma_{i}^*$ and $A(t,\epsilon) = \epsilon$ for $t\neq \epsilon$
    \item If $e\in \Sigma_a$, then $A(t,e)\in \{e,del(e)\}\Sigma_i^*$
    \item If $e\in \Sigma\setminus\Sigma_a$, then $A(t,e)\in \{e\}\Sigma_i^*$
\end{enumerate}
\end{definition}
The attack strategy $A$ defines a deterministic action based on the last event executed $e$ and modification history $t$.
We extend the function $A$ to concatenate these modifications for any string $s\in \Sigma^*$: $A(\epsilon) = A(\epsilon,\epsilon)$ and $A(s) = A(s^{|s|-1})A(A(s^{|s|-1}),s[|s|])$.
With an abuse of notation, we assume that attack strategy $A$ is encoded as a DFA $A = (X_A, \delta_A, \Sigma_m, x_{0,A})$ as in \citep{meira-goes:2020synthesis,meira-goes2021synthesistac}.
In Appendix~\ref{app:A_aut}, we show the conditions for this encoding.
Intuitively, each state of $A$ encodes one attacker's decision: insertion, deletion, or no attack. 

To identify how the attack actions affect the plant $G$ and supervisor $S$, we define projection operators to reason about events in $\Sigma_m$ in different contexts.
We define projector operator $\Pi^G$ ($\Pi^S$) that projects events in $\Sigma_m$ to events in $\Sigma$ generated by the plant (observed by the supervisor).
Formally, $\Pi^G$ outputs the event that is executed in $G$, i.e., $\Pi^G(ins(e)) = \varepsilon$ and $\Pi^G(del(e)) = \Pi^G(e) = e$. 
Similarly, $\Pi^S$ outputs the event observed by the supervisor, i.e, $\Pi^S(del(e)) = \varepsilon$ and $\Pi^S(ins(e)) = \Pi^S(e) = e$. 

\begin{example}
Figure~\ref{fig:attack-strategies} illustrates two attack strategies encoded as automata $A_1$, Fig.~\ref{fig:A1}, and $A_2$, Fig.~\ref{fig:A2}.
Since attack strategies act on observations from the plant, the encoding defines attack strategies for string $s\in \lang(A_1)$ such that $s[|s|] \in \Sigma\cup\Sigma_d$, i.e., the last event is the event observation.
For example, strings $b, a\text{ins}(b)\text{del}(b) \in \lang(A_1)$ define attack strategy $A_1(\epsilon,b)$ and $A_1(a\text{ins}(b),b)$.
After observing $s$, the attack strategy is defined for the last observed event in $s$ followed by any possible insertion in the automaton.
For instance, string $ca\in \lang(A_1)$ defines the attack strategy $A_1(c,a) = a\text{ins}(b)$ since $ca\text{ins}(b) \in \lang(A_1)$ and $\nexists e \in \Sigma_i$ such that $ca\text{ins}(b)e \in \lang(A_1)$.
Similarly, for string $aa\in \lang(A_2)$, the strategy is $A_2(a,a) = a\text{ins}(b)$.
On the other hand, for string $aa\text{ins}(b)\text{del}(b)$, the strategy is $A_2(aa\text{ins}(b),b) = \text{del}(b)$ since state $2$ does not have any feasible insertion event.

Both attack strategies $A_1$ and $A_2$ insert event $b$ when the relative distance between the vehicles is $1$.
The key difference between attackers $A_1$ and $A_2$ is when each attacker inserts event $b$. 
Attacker $A_1$ always inserts $b$ whenever the relative distance is $1$, i.e., immediately after observing $a$ the attacker inserts.
On the other hand, attacker $A_2$ inserts $b$ after the second time the relative distance is $1$.
In Fig.~\ref{fig:A2}, the attacker inserts event $b$ only after two events $a$ happened, i.e., the second time the relative distance among the cars is equal to $1$ as in Fig.~\ref{fig:M_n}. 


\end{example}

\begin{figure}[thpb]
\begin{subfigure}[t]{0.45\columnwidth}
\centering
\includegraphics[width=0.7\columnwidth]{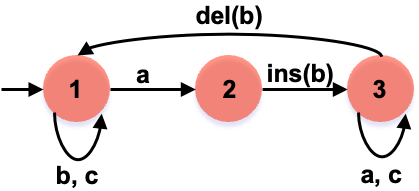}
\caption{Attacker Strategy $A_1$}
\label{fig:A1}
\end{subfigure}
\ 
\begin{subfigure}[t]{0.45\columnwidth}
\centering
\includegraphics[width=0.87\columnwidth]{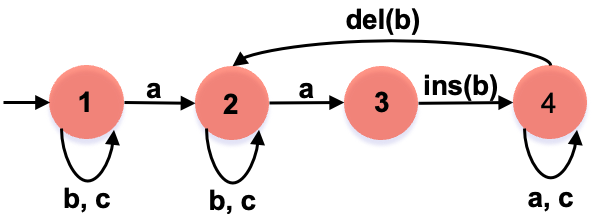}
\caption{Attacker Strategy $A_2$}
\label{fig:A2}
\end{subfigure}
\caption{Two Attack strategies $A_1$ and $A_2$}
\label{fig:attack-strategies}
\vspace{-2em}
\end{figure}

\subsection{Controlled System under Attack}

The sensor attacker disrupts the \emph{nominal} controlled system $R/G$.
A new controlled behavior is generated when the attack function $A$ is placed in the communication channel between the plant and the supervisor.
We define a new supervisor, denoted by $S_A$, that composes $S$ with the attack strategy $A$.
\begin{definition}[Attacked supervisor]
Given supervisor $S$, a set of compromised events $\Sigma_a\subseteq \Sigma$, and an attack strategy $A$. The attacked supervisor is defined for $s\in \Sigma^*$: 
\begin{equation}
S_A(s) = (S\circ \Pi^S\circ A)(s)
\end{equation}
\end{definition}
Based on $S_A$ and $G$, the closed-loop system language under attack is defined as $\mathcal{L}(S_A/G)\subseteq \lang(G)$ as in \citep{Lafortune:2021}. 
The system $S_A/G$ denotes the closed-loop system language under attack, or simply the \emph{attacked system}.
Moreover, $S_A/G$ also generates a p-language in the same manner as $R/G$.

\begin{example}
We compute the controlled language of the attacked systems $S_{A_1}/G$ and $S_{A_2}/G$, for $A_1$ and $A_2$ as shown in Fig.~\ref{fig:attack-strategies}.
Figure~\ref{fig:SA/G} depicts these two controlled systems under attack.
Let us analyze the controlled system under attack strategy $A_1$. 
At the beginning, $A_1$ does not insert any event, i.e., $A_1(\epsilon,\epsilon) = \epsilon$.
In this manner, $S_{A_1}(\epsilon) = S(\epsilon) = \{a,b,c\}$. 
For this reason, events $a,b$, and $c$ are present in the initial of the automaton in Fig.~\ref{fig:SA1/G}.

Next, let us assume that the plant executed event $a$ for which $A_1(\epsilon, a) = a\text{ins}(b)$.
Recall that after event $a$ the distance between the cars is $1$, i.e., state $1$ in $G$ in Fig.~\ref{fig:plant_G}. 
Attack strategy $A_1$ makes the supervisor return to state $2$ with $S_{A_1}(a) = S(ab) = \{a,b,c\}$, i.e., all events are enabled when the plant is in state $1$.
Thus, the controlled system under attack $A_1$ behaves exactly as in $G$ (open-loop), Fig.~\ref{fig:plant_G}.
A similar analysis can be made for the controlled system under attack $A_2$, Fig.~\ref{fig:SA2/G}.
However, the first time the plant reaches state $1$ in $G$, the controlled system behaves normally as in Fig.\ref{fig:M_n}.
After that, the controlled system under attack $A_2$ behaves as in $G$.

\end{example}
\begin{figure}[thpb]
\begin{subfigure}[t]{0.45\columnwidth}
\centering
\includegraphics[width=0.85\columnwidth]{Figs/plant_G.png}
\caption{$S_{A_1}/G$}
\label{fig:SA1/G}
\end{subfigure}
\ 
\begin{subfigure}[t]{0.45\columnwidth}
\centering
\includegraphics[width=0.87\columnwidth]{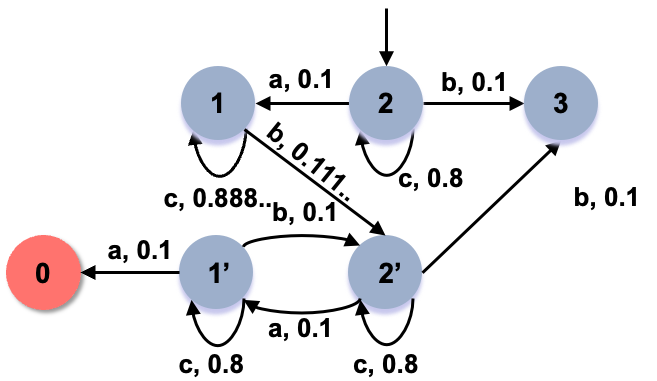}
\caption{$S_{A_2}/G$}
\label{fig:SA2/G}
\end{subfigure}
\caption{Controlled system under attack}
\label{fig:SA/G}
\vspace{-2em}
\end{figure}

\subsection{Class of Attack Strategies}

So far, we have used the general definition of attack strategies, Def.~\ref{def:attack_str}.
Herein, we have additional constraints that an attack strategy needs to satisfy.
We say the attack strategy is \emph{complete} if the attack strategy is defined for every new event observation the plant generates.
Moreover, we assume that the attack strategy is \emph{consistent} if it does not insert an event disabled by the supervisor.
Lastly, we consider \emph{successful} attack strategies as the ones that can reach the critical state, i.e., strategies that can cause damage.
\begin{definition}[Complete, Consistent, and Successful Strategies \citep{meira-goes:2021synthesis}]\label{def:complete_attacker}
An attack strategy $A$ is \emph{complete} w.r.t. $G$ and $S$ if for any $s$ in $\lang(S_A/G)$, we have that $A(s)$ is defined.
$A$ is \emph{consistent} if for any $e\in \Sigma$, $s\in \lang(S_A/G)$ such that $se \in\lang(S_A/G)$ with $A(s,e) = t$, then $S(\Pi^S(A(s)t^i))!$ and $t[i+1]\in S(\Pi^S(A(s)t^i))$ for all $i\in [|t|-1]$.
Lastly, $A$ is \emph{successful} if $\exists s\in \lang(S_A/G)$ such that $\delta_G(x_{0,G},s) = x_{crit}$.
We denote by $\Psi_A$ as the set of all complete, consistent, and successful attack strategies.  
\end{definition}

\section{Probabilistic Intrusion Detection of Sensor Deception Attacks}\label{sect:problem}
.pdfIn this section, we formulate two new problems regarding probabilistic detection of sensor attacks: the verification of $\lambda$-sensor-attack detectability and the optimal $\lambda^*$ for $\lambda$-sensor-attack detectability.
We start by formally describing the detection level of a given string and the detection language of an attack strategy. 
Following these descriptions, we define the notion of $\lambda$-sensor-attack detectability.
Next, we formulate two problems over this definition.
Lastly, we compare the definition of $\lambda$-sensor-attack detectability with the definition of $\varepsilon$-safety as in \citep{meira-goes:2020towards,Fahim2024-wodes}.

\subsection{Detection Value}

The attack detection problem is to determine if an observed behavior $s\in \lang(G)$ is generated by the nominal system $S/G$ or by an attacked system $S_A/G$.
Usually, a detection problem is described as a hypothesis-testing problem \citep{poor2013introduction}.
In our case, the null hypothesis $H_0$ is defined by the nominal system $S/G$ whereas the alternative hypothesis $H_1$ is described by an attack system $S_A/G$.

To identify from which system an observation is generated, we compare the two systems using their probabilistic language.  
Inspired by the \emph{maximum a posterior probability}, we calculate the likelihood of string $s\in \lang(G)$ by directly comparing the probability between $L_p(S/G)$ and $L_p(S_A/G)$. 
We define the \emph{detection level} of a string.

\begin{definition}[Detection level]
Let $s\in \lang(S_A/G)$, the detection level of $s$ with respect to $G$, $S$, and $A$ is:
\begin{equation}\label{eq:detection_level}
det(s) = 
\frac{L_p(S_A/G)\bigl(s \bigr)}{L_p(S_A/G)\bigl(s \bigr)+L_p(S/G)\bigl(\Pi^S(A(s))\bigr)}
\end{equation}
\end{definition}

Intuitively, $det(s)$ informs a ``detection value" of the attack strategy generating string $s\in S_A/G$.
It characterizes the likelihood of $s$ being generated by $S_A/G$ compared to $\Pi^S(A(s))$ being generated in $S/G$.
To generate $s \in \lang(S_A/G)$, the attack strategy feeds the supervisor with observation $\Pi^S(A(s))$.
For this reason, we compare the probabilities of generating $s$ in $S_A/G$ versus $\Pi^S(A(s))$ in $S/G$.

Note that if $\Pi^S(A(s))\notin \lang(S/G)$, then the value of $det(s) = 1$, i.e., the attack strategy reveals the attacker when generating $s\in\lang(S_A/G)$.
These are the \emph{only} types of attacks that logical detection systems can detect, e.g., \citep{Carvalho:2018, Lima:2019, lin2024diagnosability}.
On the other hand, the detection value is still useful when $\Pi^S(A(s))\in \lang(S/G)$.
For instance, if the attack strategy $A$ modifies the closed-loop system such that $L_p(S_A/G)(s)>L_p(S/G)(\Pi^S(A(s))$, then $det(s)>0.5$, i.e., it is more likely that the observation is coming from an attacked system instead the nominal.

\begin{example}
Let us characterize the detection value for a string in our running example with attack strategy $A_1$.
We select string $s = ca\in \lang(S_{A_1}/G)$ that has probability $L_p(S_{A_1}/G)(ca) = 0.8\times 0.1 = 0.08$ as shown in Fig.~\ref{fig:SA1/G}.
Attack strategy $A_1$ modifies $s$ to $A_1(s) = cains(b)$ as described by Fig.~\ref{fig:A1}.
The supervisor observes string $\Pi^S(cains(b)) = cab$ that has probability $L_p(S/G)(cab) = 0.8\times 0.1 \times 0.111\dots = 0.0088\cdots$.
The detection value computes the likelihood of generating string $ca$ in $S_{A_1}/G$ versus $cab$ in $S/G$.
In this case, the detection value is $det(ca) = 0.9$.
After observing $cab$, it is $90\%$ more likely that $ca$ was executed in $S_{A_1}/G$ compared to $cab$ being executed in $S/G$.
\end{example}



\subsection{Detection Language}
Although the detection value is defined for every string in $s\in \lang(S_A/G)$, only some can drive the controlled system to a critical state. 
Recall that an attack strategy is only successful if it generates a string that reaches the critical state $x_{crit}$.
Therefore, these ``successful" strings must be detected by the attack detector.
On top of detecting these strings, the attack detection must identify the attack \emph{before} the system reaches a critical state. 
In other words, we can only mitigate attacks if the attack detection detects them before it is too late, i.e., before the critical state is reached.

Based on this discussion, we define the detection language $L^A_{det}\subseteq \lang(S_A/G)$ as the strings in which the detector must make a decision otherwise it is too late. 
Recall that we assume that $x_{crit}$ is only reached via controllable events, i.e., $x \in X_G$ such that $\delta_G(x,e) = x_{crit}$ implies $e\in \Sigma_c$.
First, we define the \emph{critical language of $A$} by the strings in $\lang(S_A/G)$ that reach the critical state.
\begin{equation}
L_{crit}^A = \{s\in \lang(S_A/G)\mid \delta_G(x_{0,G},s) = x_{crit}\}
\end{equation}
As mentioned above, it is too late to detect the attack with strings in $L_{crit}^A$, i.e., the critical state has been reached.
Given plant $G$, supervisor $S$, and attack strategy $A$, we define the detection language as 
\begin{equation}\label{eq:det_lang}
L_{det}^A = \{s\in \lang(S_A/G)\mid (\exists e\in \Sigma_c.\ se\in L_{crit}^A)\wedge (\forall \sigma\in \Sigma_c,\ i<|s|.\ s^i\sigma \notin L_{crit}^A)\}
\end{equation}

A string $s$ in $L^A_{det}$ can reach the critical state with a controllable event.
Moreover, no prefix of $s$ can reach the critical state with any controllable event.
In other words, string $s$ is the shortest string to be one controllable event away from a critical state. 

\begin{example}
Let us return to our running example with attack strategy $A_1$ to investigate its detection language.
The detection language is given by $L_{det}^{A_1} = \{a, ca, cca, \dots\}$.
Let $s = a$, then $A(s) = a ins(b)$, i.e., the attack inserts event $b$ immediately after $a$ occurs. 
This insertion makes the supervisor return to state $2$ while the plant remains in state $1$ as in Figs.~\ref{fig:plant_G} and \ref{fig:sup_R}.
Since the supervisor enables event $a$ in state $2$, if the plant executes $a$, then the critical state $0$ is reached.
Therefore, the detection system must decide after observing $\Pi^S(a ins(b)) = ab$ if it should disable event $a$.
Note that since observation $ab$ belongs to $\lang(S/G)$, it cannot be detected by logical detectors.
\end{example}

\subsection{Probabilistic Sensor Detectability}
Based on the definitions of detection value, $det(s)$, and detection language, $L_{det}^A$, we define $\lambda$-sa detectability as follows:

\begin{definition}[$\lambda$-sa detectable]\label{def:lambda-sa-det}
Given plant $G$, supervisor $S$, a set of compromised events $\Sigma_a\subseteq \Sigma$, and value $\lambda \in (0.5,1]$, the controlled system is \emph{$\lambda$-sensor-attack detectable}, or simply $\lambda$-sa, if
\begin{equation}
dtc:=\inf_{A\in \Psi_A}\ \inf_{s\in L^A_{det}} det(s)\geq \lambda\label{eq:likelihood}
\end{equation}
\end{definition}

Intuitively, the controlled system is $\lambda$-sa if \emph{every} complete, consistent, and successful attack strategy significantly modifies the probability of the nominal controlled system.
The parameter $\lambda$ gives the confidence level of detection strings being more likely to be generated by the attacked systems.
For example, $0.9$-sa detectable means that detection strings for all attack strategies are at least $90\%$ more likely to have been generated in $S_A/G$ compared to their observation in $S/G$.



Based on the definition of $\lambda$-sa detectability, we formulate two problems.
First, we define a verification problem to check if a controlled system is $\lambda$-sa detectable.

\begin{problem}[Verification of $\lambda$-sa]\label{prob:ver-lsa}
Given plant $G$, supervisor $R$, a set of compromised events $\Sigma_a\subseteq \Sigma$, and $\lambda \in (0.5,1]$, verify if the controlled system is $\lambda$-sa detectable
\[dtc \geq\lambda\]
\end{problem} 

Problem~\ref{prob:ver-lsa} verifies if the controlled system is $\lambda$-sa detectable for a given $\lambda$ value.
A natural question to ask is if there exists a maximum $\lambda$ value such that the controlled system is $\lambda$-sa detectable.
Formally, the problem is posed as follows.
\begin{problem}[Maximum $\lambda$-sa]\label{prob:optimal-lsa}
Given plant $G$, supervisor $R$, and set of compromised events $\Sigma_a\subseteq \Sigma$, find, if it exists, the maximum $\lambda^*$ such that the controlled system is $\lambda^*$-sa: 
$$\lambda^* := \sup\ \{\lambda\in (0.5,1] \mid S/G \text{ is } \lambda\text{-sa detectable}\}$$
\end{problem}

\begin{remark}
In  \citep{meira-goes:2020towards,Fahim2024-wodes}, the problem of $\varepsilon$-safety is defined.
The main difference between $\varepsilon$-safety and $\lambda$-sa detectability is the $\inf_{A\in \Psi_A}$ in Eq.~\ref{eq:detection_level}.
The $\varepsilon$-safety definition only considers \emph{one} attack strategy whereas $\lambda$-sa considers all possible complete, consistent, and successful attack strategies.
The $\lambda$-sa detectability reduces to $\varepsilon$-safety when $|\Psi_A| = 1$, i.e., a single attack strategy.
\end{remark}

\section{Solution Verification of $\lambda$-sa detectability}\label{sect:solution}
.pdf
Figure~\ref{fig:overview_solution} provides an overview of our solution approach. 
First, we construct an attack system that encompasses all possible sensor attacks using the plant model $G$, supervisor $R$, and compromised event set $\Sigma_a$.
Intuitively, we construct PDES $M_n$ containing the nominal controlled behavior and $M_a$ containing all possible attacked behavior.
Next, we constructed a weighted verifier consisting of a DFA $V$ and weight function $w$.
The DFA $V$ marks the language in $L^A_{det}$ for all attacks in $\Psi_A$.
At the same time, $V$ combines the information of string executions in nominal and attack systems.
The function $w$ captures the probability ratio between executing transitions in the nominal controlled system and an attacked system.

\begin{figure}[h]
    \centering
    \includegraphics[width=1\textwidth]{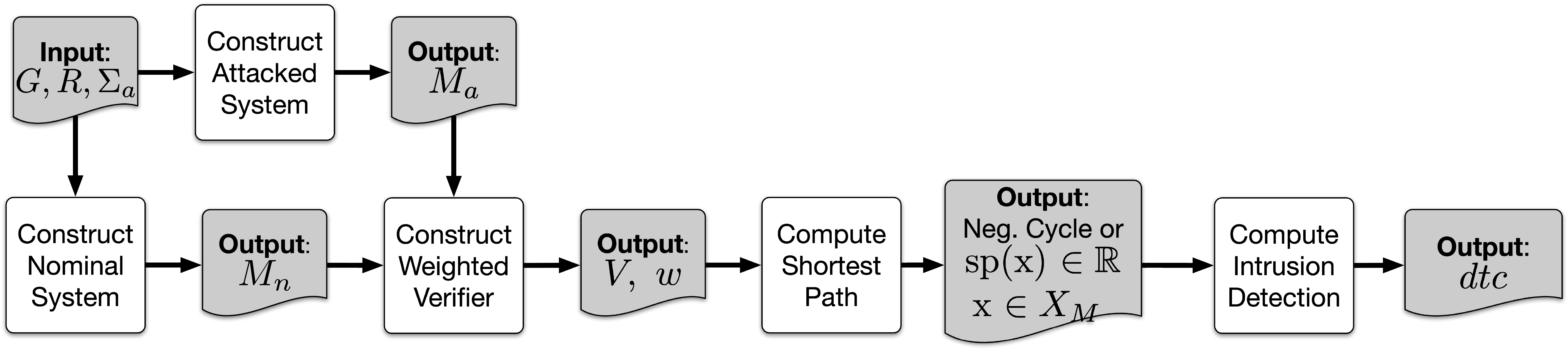} 
    \caption{Overview on solution algorithm} 
    \label{fig:overview_solution}
\end{figure}

Based on the weighted verifier, we pose a shortest path problem to identify the shortest path to reach a marked state in $V$, i.e., executing a string in a detection language $L^A_{det}$.
The shortest path problem outputs either: (1) a message that $V$ has ``negative cycles" or (2) a vector $\mathbf{sp}$ with the shortest path values from the initial state to each marked state in $V$.
Based on this output, we calculate the detection value $dtc$.
To make our approach concrete, we describe each step in detail using our running example.
Our goal is to verify if our running example is $0.9$-safe.


\subsection{Construction of nominal system}
The nominal controlled system is built as described in Section~\ref{sect:preliminaries} using the probabilistic parallel composition $||_p$.
Formally, the nominal system is defined by $M_n = R||_pG$.
Figure~\ref{fig:M_n} depicts the nominal system for our running example.

\subsection{Construction of attacked system}\label{sub:controlled-system}
The sensor attacker disrupts the \emph{nominal} controlled system $R/G$.
Based on the attack actions, we construct the attacked plant $G_a$ and the attacked supervisor $R_a$ to include \emph{every possible attack action} with respect to $\Sigma_a$ as in \citep{meira-goes:2021synthesis}. 
In this manner, we can obtain a structure that contains all possible controlled systems under sensor attacks by composing $R_a$ and $G_a$.

The attacked plant $G_a$ is a copy of $G$ with more transitions based on compromised sensors $\Sigma_a$.
Insertion events are introduced to $G_a$ as self-loops with probability $1$ since fictitious insertions do not alter the state of the plant with probability $1$.
On the other hand, deletion events are defined in $G_a$ with the same probability as their legitimate events because the attacker can only delete an event if this event has been executed in the plant $G$.
The following insertion and deletion transitions are added to $G_a$ for any $e\in \Sigma_a$.
\begin{align}
\delta_{G_a}(x,ins(e)) = x, \quad & P_{G_a}(x,ins(e),x) = 1\label{eq:ins_plant}\\
\delta_{G_a}(x,del(e)) = y, \quad & P_{G_a}(x,del(e),y) = P_{G}(x,e,y), \text{ if } \delta_G(x,e)!\label{eq:del_plant}
\end{align}
Figure~\ref{fig:Ga} shows the $G_a$ for our running example, the new transitions are highlighted in red.
Insertions in deadlock states are omitted for illustration purposes.

\begin{figure}[thpb]
\begin{subfigure}[t]{0.45\columnwidth}
\centering
\includegraphics[width=1\columnwidth]{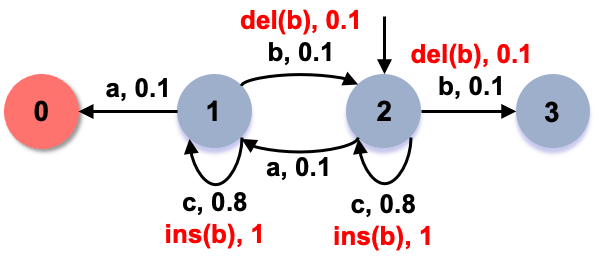}
\caption{$G_a$}
\label{fig:Ga}
\end{subfigure}
\ 
\begin{subfigure}[t]{0.45\columnwidth}
\centering
\includegraphics[width=0.71\columnwidth]{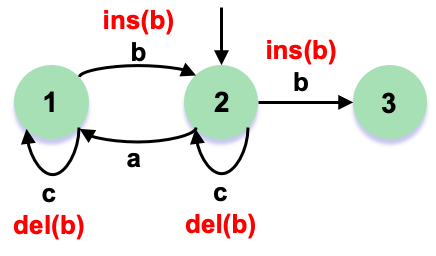}
\caption{$R_a$}
\label{fig:Ra}
\end{subfigure}
\\
\begin{subfigure}[t]{1\columnwidth}
\centering
\includegraphics[width=0.55\columnwidth]{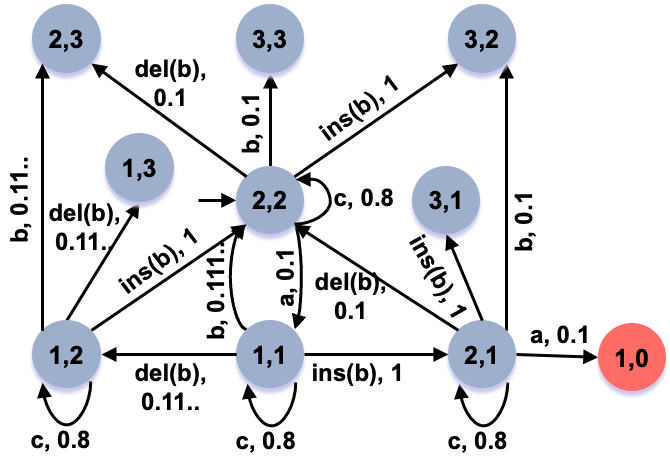}
\caption{$M_a$}
\label{fig:Ma}
\end{subfigure}
\caption{Attacked plant, supervisor, and system}
\label{fig:attacked-models}
\vspace{-2em}
\end{figure}


Similar to the construction of $G_a$, we define the attacked supervisor $R_a$ as a DFA.
For the supervisor, insertions are observed as legitimate events while deletions do not change the state of the supervisor.
For any event $e\in \Sigma_a$, the following transitions are added to $R_a$ on top of those already in $R$.
\begin{align}
\delta_{R_a}(x,ins(e)) &= \delta_R(x,e) \text{ if } \delta_R(x,e)! \label{eq:ins_sup}\\
\delta_{R_a}(x,del(e)) &= x \text{ if } \delta_R(x,e)! \label{eq:del_sup}
\end{align}
Figure~\ref{fig:Ra} shows the attacked supervisor for our running example.

All possible controlled systems under attack are represented by $M_a = R_a||_p G_a$.
Figure~\ref{fig:Ma} shows the system $M_a$ for our running example.
Recall that states in $M_a$ are of the format $(x_{R},x_{G})$ for $x_R\in X_{R_a}$ and $x_G\in X_{G_a}$.
Note that the attacker can cause a mismatch between these states which was not possible in the nominal $M_n$ as in Fig.~\ref{fig:M_n}.
Although this model is simple, it is well-suited to the problem we are investigating since we want to detect every possible attack strategy. 
Next, we have two propositions that link the language of $M_a$ with languages generated by attacked systems $\lang(S_A/G)$ for a given $A$.
The first proposition shows that a complete and consistent attack strategy generates a behavior in $M_a$.
\begin{proposition}\label{prop:Sa-Ma}
Let attack strategy $A$ be complete and consistent.
For every $s\in \lang(S_A/G)$, then $A(s) \in \lang(M_a)$.
\end{proposition}
\begin{proof}
It follows by the construction of $G_a$, $R_a$, and by $A$ being a complete and consistent attack strategy.
\end{proof}

Next, every string in $M_a$ can be generated by a complete and consistent attack strategy.
\begin{proposition}\label{prop:Ma-Sa}
For every $s\in \lang(M_a)$, there exists $A$ complete and consistent such that $\Pi^G(s)\in \lang(S_A/G)$.
\end{proposition}
\begin{proof}
We start by showing $\Rightarrow$ by constructing an attack strategy $A$ that generates $s\in \lang(M_a)$.
First, we can break $s$ into $0\leq k\leq|s|$ substrings such that $s = t_1\dots t_k$.
Moreover, each substring satisfies: $t_1\in \Sigma_i^*$ and $t_i\in (\Sigma\cup\Sigma_d)\Sigma_i^*$ for $1<i\leq k$.

Intuitively, we are breaking $s$ into $k$ substrings to be generated by the attack strategy $A$.
Recall that $A$ needs to satisfy the conditions in Def.~\ref{def:attack_str}.
Attack $A$ will output each of these $t_i$, e.g., $A(\epsilon,\epsilon) = t_1$, and $A(t_1,\mask(t_2[1])) = t_2$.
We construct $A$ for $s$ as follows:
\begin{align*}
A(\epsilon,\epsilon) &= t_1\\
A(t_1\dots t_j,\mask(t_{j+1}[1])) &= t_{j+1} \forall\ 1<j\leq k-1
\end{align*}
For other strings $t\in\Sigma_m^*\setminus\{\epsilon\}$ and event $e\in\Sigma$, the attack strategy is $A(t,e)=e$.
Attack strategy $A$ is complete and consistent by construction.

Now, we show that $\Pi^G(s)\in \lang(S_A/G)$ by showing that $\Pi^G(t_1\dots t_j)\in \lang(S_A/G)$.
We show this by recursively showing that $te\in \lang(G)$, $t\in \lang(S_A/G)$ and $e\in S_A(t)$ which implies that $te\in \lang(S_A/G)$.  
By definition of $\lang(S_A/G)$, $\epsilon = \Pi^G(t_1)\in \lang(S_A/G)$.

By construction of $t_1$ and $t_2$, we have $\Pi^G(t_1t_2) = \mask(t_2[1])$.
Now by construction of $R_a$, it follows that $x_R = \delta_{R_a}(x_{0,R_a},t_1) = \delta_R(x_{0,R},\Pi^S(t_1))$. 
Since $t_1t_2 \in \lang(M_a)$ and the definition of $R_a||_pG_a$, we have that $t_2[1]\in \Gamma_{R_a}(x_R)$ and $\mask(t_2[1])\in \Gamma_{R}(x_R)$.
Thus, the event $\mask(t_2[1])$ is allowed by $S(A(\epsilon,t_1)) = S_A(\epsilon)$.
As $t_1t_2[1]\in \lang(G_a)$, then $\Pi^G(t_1t_2) = \mask(t_2[1])\in \lang(G)$ by construction of $G_a$.
In summary, we have $\epsilon\in \lang(S_A/G)$, $\Pi^G(t_1t_2) = t_2[1]\in \lang(G)$, and $t_2[1]\in S_A(t_1)$, which implies that $\Pi^G(t_1t_2)\in \lang(S_A/G)$.
By similar recursive arguments, we can show that $\Pi^G(t_1\dots t_j) = t_2[1]\dots t_j[1]\in \lang(S_A/G)$ for any $1<j\leq k$.
\end{proof}

\subsection{Constructing the Weighted Verifier}
Once $M_a = R_a||_p G_a$ is constructed, we identify all possible languages $L_{det}^A$ for any $A\in \Pi_A$.
Recall that $L_{det}^A$ is defined by the shortest strings that are one controllable event away from a critical state.
Based on $M_a$, we define the detection states as follows:
\begin{equation}
X_{det} = \{(x_R,x_G)\in X_{M_a}\mid \exists e\in (\Sigma_d\cap\Sigma)\text{ s.t. } (e \in \Gamma_{R_a}(x_R))\wedge (\delta_{G_a}(x_{G},e)=x_{crit}\}
\end{equation}
In the $M_a$ in Fig.~\ref{fig:Ma}, the detection state is $(2,1)$ since the critical state $(1,0)$ is reached via controllable event $a$. 
The detection states $X_{det}$ are related to $L_{det}^A$ since they are states one controllable event away from a critical state.


Next, we need a structure where we can directly compare string executions in $S_A/G$ versus $S/G$.
Inspired by the verifier automaton in \citep{yoo2002polynomial}, we define the weighted verifier, DFA $V$, and weight function $w$.
The verifier automaton $V$ marks the strings in $L^A_{det}$ for any $A\in \Psi_A$.
Moreover, weights $w$ contain the probability information of executing transitions in attacked systems $S_A/G$ and nominal system $S/G$.
We start by constructing the verifier $V$ similarly to the steps described in \citep{meira-goes:2020towards}. 

\begin{definition}\label{def:verifier}
Given $M_n$, $M_a$, and the detection states $X_{det}$, we define verifier $V$ as: 
(1) $X_{V}\subseteq X_{R}\times X_G \times X_{R_a}\times X_{G_a}$; (2) $x_{0,V} = (x_{0,R},x_{0,G},x_{0,R_a},x_{0,G_a})$; (3) $\delta_{V}((x_1,x_2,x_3,x_4),e) = (y_1,y_2,y_3,y_4)$ if $\delta_{R,G}((x_1,x_2),\Pi^S(e)) = (y_1,y_2)$ and $\delta_{R_a,G_a}((x_3,x_4),e) = (y_3,y_4)$  for $e\in \Sigma_m$, $x_1,y_1\in X_R$, $x_2,y_2\in X_G$, $x_3,y_3\in X_{R_a}$, and $x_4,y_4\in X_{G_a}$ with $(x_3,x_4)\not\in X_{det}$, otherwise is undefined; and (4) $X_{m,V} = \{(x_1,x_2,x_3,x_4) \mid \ (x_3,x_4)\in X_{det}\}$.
\end{definition}

With abuse of notation, we only describe $V$ by its accessible and co-accessible parts, i.e., the $Trim(V)$ operator as in \citep{Lafortune:2021} is applied after Def.~\ref{def:verifier}.
Figure~\ref{fig:verifier} depicts the verifier $V$ constructed based on Def.~\ref{fig:verifier}.
To construct this automaton, we start from the initial state $x_{0,V}$ and perform a reachability analysis to obtain the next states via $\delta_V$.
For example from state $(2,2,2,2)$ and event $a$, state $(1,1,1,1)$ is reached.
In this scenario, $ego$ moves one cell closer to $adv$ in both systems.
From state $(1,1,1,1)$ and event $ins(b)$, state $(2,2,2,1)$ is reached.
In this case, the nominal system moves to state $(2,2)$ since the insertion $ins(b)$ is observed as event of $b$.
However, the attacked system moves to state $(2,1)$ where the relative distance remains $1$.
State $(2,2,2,1)$ is a marked state since state $(2,1)$ is a detection state.
Next, we discuss the weights in the verifier.

\begin{figure}[thpb]
\centering
\includegraphics[width=0.65\columnwidth]{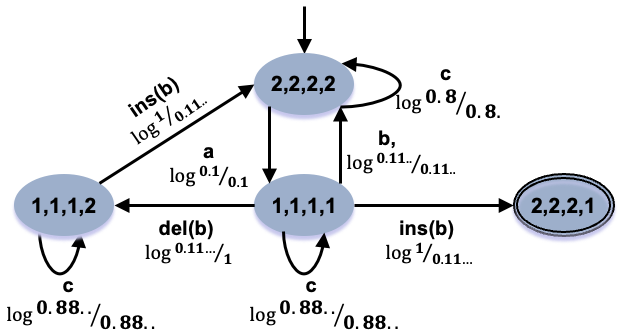}
\caption{Weighted verifier $V$}
\label{fig:verifier}
\end{figure}

We define weights for verifier $V$ based on the transition probabilities in $M_n$ and $M_a$.
For example, the transition in $V$ from state $(1,1,1,1)$ to state $(2,2,2,1)$ via event $ins(b)$ captures the information of the execution in $M_n$ and in $M_a$.
In the case of $M_n$, the nominal system observes event $b$, which has a probability of $0.111\dots$, Fig.~\ref{fig:M_n}.
In the attacked system, the probability of executing $ins(b)$ is equal to $1$, Fig.~\ref{fig:Ma}.
Thus, the ratio of executing this transition in the attacked system is $9$ times more likely than executing in the nominal system, i.e., $1/0.111\dots = 9$.
We use the logarithm of this ratio as a weight for this transition as shown in Fig.~\ref{fig:verifier}.
The use of the logarithm will become clear when we pose the shortest path problem.

Formally, we define the weight function $w:X_V\times \Sigma_m\times X_V\rightarrow \mathbb{R}$ for each transition $\delta_V((x_1,x_2,x_3,x_4),e) = (y_1,y_2,y_3,y_4)$ as:

\begin{equation}\label{eq:ratio}
w((x_1,x_2,x_3,x_4),e,(y_1,y_2,y_3,y_4)) = \log \frac{P_{M_a}((x_3,x_4),e,(y_3,y_4)}{P_{M_n}((x_1,x_2),\Pi^S(e),(y_1,y_2))}
\end{equation}
In Eq.~\ref{eq:ratio}, we have that $P_{M_n}((x_1,x_2),\epsilon,(y_1,y_2)) = 1$, i.e., the probability of executing the empty string is always one.

Based on the verifier $V$, we calculate the weights as shown in Fig.~\ref{fig:verifier} using Eq.~\ref{eq:ratio}.
For example, the weight for transition $\delta_V((2,2,2,2),b) = (1,1,1,1)$ is equal to $\log{(1)}=0$ since $P_{M_a}((2,2),b,(1,1)) = P_{M_n}((2,2),b,(1,1)) = 0.1$.

Propositions~\ref{prop:Sa-Ma} and~\ref{prop:Ma-Sa} have linked the strings in $M_a$ to strings in $S_A/G$ for an attack strategy.
Next, we show that the strings in $\lang_m(V)$ are related to strings in a detection language $L_A^{det}$.

\begin{proposition}\label{prop:det_lang}
A string $s\in \lang_m(V)$ if and only if $\exists A\in \Psi_A$ such that $\Pi^G(s)\in L_{det}^A$.
\end{proposition}
\begin{proof}
This proposition follows from Propositions~\ref{prop:Sa-Ma} and~\ref{prop:Ma-Sa} and the construction of $V$ and its detection states $X_{det}$    
\end{proof}

\subsection{Finding the Shortest Path}
The verifier automaton $V$ has the information to find the attack strategy $A$ and the string in $L^A_{det}$ with the smallest intrusion detection value, $dtc$, as in Def.~\ref{def:lambda-sa-det}.
This string is related to the shortest path from the initial state to any detection state, i.e., the string executed by the path.
To show this relationship, we present the relationship between the weight of a path in $V$ and the probabilities of executing this path in $M_n$ and $M_a$

Let us select the path $p = (2,2,2,2)a(1,1,1,1)ins(b)\allowbreak(2,2,2,1)$ from the verifier $V$ in Fig.~\ref{fig:verifier}.
This path is generated using attack strategy $A_1$ in Fig.~\ref{fig:A1}.
The weight of this path is the sum of each transition weight: $$\log(\frac{0.1}{0.1})+\log(\frac{1}{0.11\dots})=\log(1)+\log(9) = \log(1\times 9) = \log(9)$$
In Eq.~\ref{eq:ratio}, we define transition weights as the ratio of transition probabilities in $M_a$ over  $M_n$.
We rewrite the weight of path $p$:
\begin{align*}
\log(\frac{0.1}{0.1})+\log(\frac{1}{0.11\dots}) &= \log(\frac{0.1}{0.011\dots})\\&=\log(\frac{L_p(S_{A_1}/G)(a)}{L_p(S/G)(ab)})\\& =\log(9)
\end{align*}
Executing $p$ in $M_a$ is $9$ times more likely than executing in $M_n$.

By finding the shortest path in $V$, we find the smallest ratio of executing a path in $M_a$ over executing in $M_n$.
The shortest path is defined as:
\begin{definition}[Shortest path]\label{def:shortest-path}
Given verifier $V$ with weight function $w$, the shortest marked path is defined as the path with the shortest sum of weights from the initial state to a marked state in $V$:
$$\inf_{\rho:=x_0e_0\dots x_{|\rho|}\in Paths_m(V)}\sum_{i=0}^{i<|\rho|}w(x_i,e_i,x_{i+1})$$
where $Path_m(G)$ is a path starting in $x_{0,G}$ and ending in $X_{m,G}$, $Path_m(G) = \{x_0e_0\dots x_n\in (X_G\times \Sigma)^*X_G \mid  x_0 = x_{0,G}\wedge x_{i+1} = \delta_G(x_i,e_i), i<n \wedge x_n\in X_{m,G}\}$.
\end{definition}

The following proposition formally ties the weight of a path in $V$ with the probability ration of executing a string in $M_a$ and $M_n$.
\begin{proposition}\label{prop:weight-Ma}
Given verifier $V$ and weight function $w$, for any $s\in \lang_m(V)$ with path $x_0s[1]x_1\dots s[|s|]x_{|s|}$ then
$$\sum_{i=0}^{|s|-1} w(x_i,s[i+1],x_{i+1}) = \log \left(\frac{L_p(S_A/G)(\Pi^G(s))}{L_p(S/G)(\Pi^S(s))} \right)$$
where attack strategy $A$ is constructed as in the proof of Prop.~\ref{prop:Ma-Sa}.
\end{proposition}
\begin{proof}
This result follows by the definition of the verifier, the weight function, and the logarithm property of multiplication, i.e., $log(ab) = log(a) +log(b)$.
Recall that each verifier state $x_i$ is defined by $(x_{i,1},x_{i,2},x_{i,3},x_{i,4})$ (Def.~\ref{def:verifier}).
It follows that:
\begin{align*}
\sum_{i=0}^{|s|-1}w(x_i,s[i+1],x_{i+1}) &= \sum_{i=0}^{|s|-1} \log\left(\frac{P_{M_a}((x_{i,3},x_{i,4}),s[i+1],(x_{i+1,3},x_{i+1,4})}{P_{M_n}((x_{i,1},x_{i,2}),\Pi^S(s[i+1]),(x_{i+1,1},x_{i+1,2}))}\right)\\
&= \log\left(\Pi_{i=0}^{|s|-1} \frac{P_{M_a}((x_{i,3},x_{i,4}),s[i+1],(x_{i+1,3},x_{i+1,4})}{P_{M_n}((x_{i,1},x_{i,2}),\Pi^S(s[i+1]),(x_{i+1,1},x_{i+1,2}))}\right)
\end{align*}
By Prop.~\ref{prop:Ma-Sa}, we can construct an attack strategy $A$ such that $\Pi^G(s)\in \lang(S_A/G)$.
Moreover, by construction of $M_a$, we have that $\Pi_{i=0}^{|s|-1} P_{M_a}((x_{i,3},x_{i,4}),s[i+1],(x_{i+1,3},x_{i+1,4}) = L_p(S_A/G)(\Pi^G(s))$.
And by construction of $M_n$, we have $\Pi_{i=0}^{|s|-1} P_{M_n}((x_{i,1},x_{i,2}),\Pi^S(s[i+1]),(x_{i+1,1},x_{i+1,2}) = L_p(S/G)(\Pi^S(s))$.
Therefore, we have:
\begin{align*}
\sum_{i=0}^{|s|-1}w(x_i,s[i+1],x_{i+1})  = \log\left(\frac{L_p(S_A/G)(\Pi^G(s))}{L_p(S/G)(\Pi^S(s))}\right)
\end{align*}
\end{proof}

The shortest path problem is a well-known problem in graph theory with polynomial-time algorithm solutions, e.g., Bellman-Ford, and Dijkstra's algorithms \citep{cormen2022introduction}.
Since $V$ can have negative and positive weights, we use the Bellman-Ford algorithm to compute the shortest path in $V$.
The Bellman-Ford algorithm outputs either: (1) a vector, $\mathbf{sp} \in \mathbb{R}^{|X_{m,V}|}$, storing the smallest real values for paths from the initial state to marked states, or (2) an output saying that the graph has a ``negative cycle."

Back to our running example, we run the Bellman-Ford algorithm using the weighted verifier $V$ depicted in Fig.~\ref{fig:verifier}.
The verifier does not have negative cycles since all cycles have $0$ weight, i.e., $\log(1) = 0$.
The shortest path returns the vector $\mathbf{sp} = [\log(9)]$.

\subsection{Extracting the Intrusion Detection Value}
Solving the shortest path problem for the weighted verifier $V$ gives us the information to solve Problems~\ref{prob:ver-lsa} and~\ref{prob:optimal-lsa}.
Herein, we describe the solution for Problem~\ref{prob:ver-lsa}.
The solution for Problem~\ref{prob:optimal-lsa} follows the same steps.

The shortest path problem can return two possible outputs: (1) a vector, $\mathbf{sp}\in \mathbb{R}^{|X_{m,V}|}$, or (2) ``negative cycle."
In the case of output (2), Problem~\ref{prob:ver-lsa} returns that $M_n$ is not $\lambda$-sa detectable.
The shortest path can be arbitrarily small when $V$ has a negative cycle.
We can select a string in $M_a$ that reaches a marked state that its probability of execution is much smaller than in $M_n$, i.e., $L_p(M_a)(s)<< L_p(M_n)(\Pi^S(s))$.
In other words, we can construct an attacker $A$ as in Prop.~\ref{prop:Ma-Sa} that generates this string with an arbitrarily small detection value. 
In this scenario, there exists an attacker that hides its probabilistic trace to be negligible.

\begin{theorem}\label{theo:negative-cycle}
Let $M_n$, $M_a$, $\lambda \in (0.5,1]$ be given.
If the Bellman-Ford algorithm returns that the weighted verifier $V$ has a negative cycle, then $M_n$ is not $\lambda$-sa detectable.
\end{theorem}

\proof
When $V$ has a negative cycle, the shortest path has a limit as $-\infty$.
Thus, we can find a path $s\in \lang_m(V)$ with path $\rho = x_0s[1]x_1\dots s[|s|]x_{|s|+1}$ with weight smaller than $\log(\frac{\lambda}{1-\lambda})$.
$$\sum_{i=0}^{i<|s|}w(x_i,s[i],x_{i+1})< \log\left(\frac{\lambda}{1-\lambda}\right)$$
Using Prop.~\ref{prop:weight-Ma}, we have:
\begin{align*}
\log\left(\frac{L_p(S_A/G)(s)}{L_p(S/G)(\Pi^S(s))}\right)< \log\left(\frac{\lambda}{1-\lambda}\right)\\
\frac{L_p(S_A/G)(\Pi^G(s))}{L_p(S/G)(\Pi^S(s))}<\left(\frac{\lambda}{1-\lambda}\right)
\end{align*}
Manipulating the equation above, we get:
\begin{align*}
\frac{L_p(S_A/G)(\Pi^G(s))}{L_p(S_A/G)(\Pi^G(s))+L_p(S/G)(\Pi^S(s))}< \lambda 
\end{align*}
Since $s\in \lang_m(V)$ and by Prop.~\ref{prop:det_lang}, the system $M_n$ is not $\lambda$-sa detectable.
\endproof

In the case the Bellman-Ford algorithm returns the vector $\mathbf{sp}\in \mathbb{R}^{|X_{m,V}|}$, we can find the value $dtc$ in Eq.~\ref{eq:likelihood}.
Depending on this calculated value, we return (i) $M_n$ is $\lambda$-sa detectable, or (ii) $M_n$ is \textbf{not} $\lambda$-sa detectable.
We calculate the value $dtc$ with the lowest value $val$ in the vector $\mathbf{sp}$, i.e., $val = \min \textbf{sp}$.
\begin{align}
\frac{L_p(S_A/G)(\Pi^G(s))}{L_p(S/G)(\Pi^S(s))} = \exp(val)\nonumber\\
L_p(S_A/G)(\Pi^G(s)) = \exp(val)L_p(S/G)(\Pi^S(s))\label{eq:relation-ratio}
\end{align}
Where $s$ is the string in $V$ generated by the shortest path value obtained by the Bellman-Ford algorithm.
Using the relation in Eq.~\ref{eq:relation-ratio}, Eq.~\ref{eq:likelihood} gives us the value of $dtc$.
\begin{align}
dtc &= \frac{\exp(val)L_p(S/G)(\Pi^S(s))}{L_p(S/G)(\Pi^S(s))(1+\exp(val))}\nonumber \\
& =\frac{exp(val)}{1+exp(val)} \label{eq:final-int}
\end{align}

Therefore, the shortest path value $val$ to a marked state in $V$ provides us the $dtc$ as in Def.~\ref{def:lambda-sa-det}.
Formally, we have the following theorem.

\begin{theorem}\label{theo:vector-shortest}
Let $M_n$, $M_a$, $\lambda\in (0.5,1]$ be given.
Also, let the Bellman-Ford algorithm for the weighted verifier $V$ return vector $\mathbf{sp}\in \mathbb{R}^{|X_{m,V}|}$ with $val$ being the lowest value in $\mathbf{sp}$.
The system $M_n$ is $\lambda$-sa detectable if and only if 
\begin{equation}\label{eq:thm}
\frac{exp(val)}{1+exp(val)}\geq \lambda
\end{equation}
\end{theorem}

\proof \quad 
We start by proving $\Rightarrow$: If $M_n$ is $\lambda$-sa detectable, then Eq.~\ref{eq:thm} holds.
First, we assume that $M_n$ is  $\lambda$-sa detectable, which implies that $dtc = \inf_{A\in\Psi_A}\inf_{s\in L^A_{det}} det(s) \geq \lambda$.
Therefore, there exists an attack strategy $A\in \Psi_A$ and string $s\in L_{det}^A$ such that for all other $A'\in \Psi_A$ and $t\in L_{det}^{A'}$: $det(s)\leq det(t)$.
Let $A$ and $s$ be the attack strategy and string that satisfy the proposition above.
Using Prop.~\ref{prop:Sa-Ma}, we can find a string $u = A(s)\in \lang_m(V)$ such that $\Pi^G(u) = s$.
By Prop.~\ref{prop:weight-Ma}, we have: 
\begin{equation}\label{eq:w_u}
w_u = \exp\left(\sum_{i=0}^{|u|-1} w(x_i,u[i+1],x_{i+1})\right)  = \frac{L_p(S_A/G)(s)}{L_p(S/G)(\Pi^S(s))}
\end{equation}
where $u$ generates path $x_0u[1]\dots u[|u|]x_{|u|}$ in $V$.
Similar to Eq.~\ref{eq:final-int}, we can write the detection value of $s$ as $det(s) = \frac{w_u}{1+w_u}$.
Since we assumed that $M_n$ is $\lambda$-sa detectable $\frac{w_u}{1+w_u}\geq \lambda$.
It remains to be proved that $w_u = \exp(val)$.

We show that $w_u = \exp(val)$ by contradiction.
By the solution of the Bellman-Ford algorithm, we can construct a string $v\in \lang_m(V)$ such that $w_v = \exp(val)$, where $w_v$ is the weight of string $v$ calculated as in Eq.~\ref{eq:w_u}
Assume that $w_u\neq \exp(val)$.
If $w_u<\exp(val)$, we found a new path in $V$ with a weight smaller than the shortest path $val$, i.e., a contradiction.
If $w_u>\exp(val)$, we found a string $v\in \lang_m(V)$ with weight smaller than $w_u$.
Using Prop.~\ref{prop:weight-Ma}, we can calculate $det(\Pi^G(v)) = \frac{w_v}{w_v+1}$.
Since $w_u>w_v$, $w_u>0$, and $w_v>0$, it follows: 
\begin{equation}\label{eq:ineq-cont}
det(s) = det(\Pi^G(u)) = \frac{w_u}{w_u+1}> \frac{w_v}{w_v+1}>det(\Pi^G(v))
\end{equation}
However, Eq.~\ref{eq:ineq-cont} contradicts that $det(s)$ is the smallest detection value as we assumed above.
Therefore, it must be that $w_u = w_v = \exp(val)$ which concludes our $\Rightarrow$ proof.

Next, we need to prove $\Leftarrow$: If Eq.~\ref{eq:thm} holds, then $M_n$ is $\lambda$-sa detectable.
We can prove this statement by contradiction.
Let us assume that Eq.~\ref{eq:thm} holds and that $M_n$ \emph{is not} $M_n$ $\lambda$-sa detectable.
Since $M_n$ is not $\lambda$-sa detectable, then there exist an attack strategy $A$ and string $s\in L_{det}^A$ such that $det(s)<\lambda$.
Using Prop.~\ref{prop:Sa-Ma}, we can find an string $t\in \lang_m(V)$ such that $A(s) = t$ and $det(s) = \frac{w_t}{1+w_t}$ where $w_t$ is calculated as in Eq.~\ref{eq:w_u}.
We show that $w_t$ will be smaller than $exp(val)$, which contradicts the assumption that $val$ is the smallest weight in $V$.

Let $v\in \lang_m(V)$ such that $w_v = exp(val)$, where $w_v$ is the weight of string $v$ similar to Eq.\ref{eq:w_u}.
With $w_t$, it follows that $det(\Pi^G(v)) = \frac{w_v}{w_v+1}>\lambda$ by our assumption that Eq.~\ref{eq:thm} holds.
Using $det(s)<\lambda$ and $det(\Pi^G(v))>\lambda$, it follows that  $\frac{w_t}{1+w_t}<\frac{w_v}{w_v+1}$.
The last inequality implies that $w_t<w_v = \exp(val)$ since $w_v>0$ and $w_t>0$, i.e., string $t$ has weight less than $val$.
However, $val$ is the smallest weight in $V$, i.e., a contradiction.
Therefore, if Eq.~\ref{eq:thm} holds, then $M_n$ is $\lambda$-sa detectable.
This concludes our proof.
\endproof

Back to our running example, the shortest path from the Bellman-Ford algorithm for verifier in Fig.~\ref{fig:verifier} is $\mathbf{sp} = \log(9)$.
Since there is a single marked state, $val = \log(9)$.
Following Thm.~\ref{theo:vector-shortest}, we have:
\begin{equation}
\frac{\exp(val)}{1+\exp(val)}=0.9\geq 0.9
\end{equation}
Thus, the system $M_n$ is $0.9$-safe.

\subsection{Complexity of $\lambda$-sa detectability}
Our last result is related to the complexity of solving Problems~\ref{prob:ver-lsa} and~\ref{prob:optimal-lsa} due to the Bellman-Ford algorithm.
\begin{theorem}\label{theo:complexity}
Solving Problem~\ref{prob:ver-lsa} has worst-case time-complexity of $O(|X_V|^2) = O(|X_G\times X_R|^4)$.
\end{theorem}
\begin{proof}
This result follows by the complexity of the Bellman-Ford algorithm and the construction of $V$.
\end{proof}

\section{Conclusion}\label{sect:conclusion}
This paper investigated a new sensor attack detection notion using probabilistic information. 
We proposed the notion of $\lambda$-sa detectability that ensures probabilistic detection for all complete, consistent, and successful sensor attack strategies.
This notion generalizes the previously defined $\epsilon$-safety, introduced in \citep{meira-goes:2020towards,Fahim2024-wodes}, by considering general classes of sensor attack strategies instead of a single strategy as in $\epsilon$-safety.
We show that $\lambda$-sa detectability can be verified in polynomial time by reducing our verification problem to a shortest path problem in a graph.
We leave for future work considering the detection of both sensor and actuator attacks. 
\bibliography{romulo}


\begin{thebibliography}{37}
\ifx \bisbn   \undefined \def \bisbn  #1{ISBN #1}\fi
\ifx \binits  \undefined \def \binits#1{#1}\fi
\ifx \bauthor  \undefined \def \bauthor#1{#1}\fi
\ifx \batitle  \undefined \def \batitle#1{#1}\fi
\ifx \bjtitle  \undefined \def \bjtitle#1{#1}\fi
\ifx \bvolume  \undefined \def \bvolume#1{\textbf{#1}}\fi
\ifx \byear  \undefined \def \byear#1{#1}\fi
\ifx \bissue  \undefined \def \bissue#1{#1}\fi
\ifx \bfpage  \undefined \def \bfpage#1{#1}\fi
\ifx \blpage  \undefined \def \blpage #1{#1}\fi
\ifx \burl  \undefined \def \burl#1{\textsf{#1}}\fi
\ifx \doiurl  \undefined \def \doiurl#1{\url{https://doi.org/#1}}\fi
\ifx \betal  \undefined \def \betal{\textit{et al.}}\fi
\ifx \binstitute  \undefined \def \binstitute#1{#1}\fi
\ifx \binstitutionaled  \undefined \def \binstitutionaled#1{#1}\fi
\ifx \bctitle  \undefined \def \bctitle#1{#1}\fi
\ifx \beditor  \undefined \def \beditor#1{#1}\fi
\ifx \bpublisher  \undefined \def \bpublisher#1{#1}\fi
\ifx \bbtitle  \undefined \def \bbtitle#1{#1}\fi
\ifx \bedition  \undefined \def \bedition#1{#1}\fi
\ifx \bseriesno  \undefined \def \bseriesno#1{#1}\fi
\ifx \blocation  \undefined \def \blocation#1{#1}\fi
\ifx \bsertitle  \undefined \def \bsertitle#1{#1}\fi
\ifx \bsnm \undefined \def \bsnm#1{#1}\fi
\ifx \bsuffix \undefined \def \bsuffix#1{#1}\fi
\ifx \bparticle \undefined \def \bparticle#1{#1}\fi
\ifx \barticle \undefined \def \barticle#1{#1}\fi
\bibcommenthead
\ifx \bconfdate \undefined \def \bconfdate #1{#1}\fi
\ifx \botherref \undefined \def \botherref #1{#1}\fi
\ifx \url \undefined \def \url#1{\textsf{#1}}\fi
\ifx \bchapter \undefined \def \bchapter#1{#1}\fi
\ifx \bbook \undefined \def \bbook#1{#1}\fi
\ifx \bcomment \undefined \def \bcomment#1{#1}\fi
\ifx \oauthor \undefined \def \oauthor#1{#1}\fi
\ifx \citeauthoryear \undefined \def \citeauthoryear#1{#1}\fi
\ifx \endbibitem  \undefined \def \endbibitem {}\fi
\ifx \bconflocation  \undefined \def \bconflocation#1{#1}\fi
\ifx \arxivurl  \undefined \def \arxivurl#1{\textsf{#1}}\fi
\csname PreBibitemsHook\endcsname

\bibitem[\protect\citeauthoryear{Allg\"ower et~al.}{2019}]{allgower2019}
\begin{barticle}
\bauthor{\bsnm{Allg\"ower}, \binits{F.}},
\bauthor{\bsnm{{Borges de Sousa}}, \binits{J.}},
\bauthor{\bsnm{Kapinski}, \binits{J.}},
\bauthor{\bsnm{Mosterman}, \binits{P.}},
\bauthor{\bsnm{Oehlerking}, \binits{J.}},
\bauthor{\bsnm{Panciatici}, \binits{P.}},
\bauthor{\bsnm{Prandini}, \binits{M.}},
\bauthor{\bsnm{Rajhans}, \binits{A.}},
\bauthor{\bsnm{Tabuada}, \binits{P.}},
\bauthor{\bsnm{Wenzelburger}, \binits{P.}}:
\batitle{Position paper on the challenges posed by modern applications to cyber-physical systems theory}.
\bjtitle{Nonlinear Analysis: Hybrid Systems}
\bvolume{34},
\bfpage{147}--\blpage{165}
(\byear{2019})
\end{barticle}
\endbibitem

\bibitem[\protect\citeauthoryear{Cassandras and Lafortune}{2021}]{Lafortune:2021}
\begin{bbook}
\bauthor{\bsnm{Cassandras}, \binits{C.G.}},
\bauthor{\bsnm{Lafortune}, \binits{S.}}:
\bbtitle{Introduction to Discrete Event Systems},
\bedition{3rd} edn.
\bpublisher{Springer},
\blocation{AG, Switzerland}
(\byear{2021})
\end{bbook}
\endbibitem

\bibitem[\protect\citeauthoryear{Cormen et~al.}{2022}]{cormen2022introduction}
\begin{bbook}
\bauthor{\bsnm{Cormen}, \binits{T.H.}},
\bauthor{\bsnm{Leiserson}, \binits{C.E.}},
\bauthor{\bsnm{Rivest}, \binits{R.L.}},
\bauthor{\bsnm{Stein}, \binits{C.}}:
\bbtitle{Introduction to Algorithms}.
\bpublisher{MIT press},
\blocation{Cambridge, MA, United States of America}
(\byear{2022})
\end{bbook}
\endbibitem

\bibitem[\protect\citeauthoryear{Cho and Marcus}{1989}]{Cho:1989}
\begin{barticle}
\bauthor{\bsnm{Cho}, \binits{H.}},
\bauthor{\bsnm{Marcus}, \binits{S.I.}}:
\batitle{On supremal languages of classes of sublanguages that arise in supervisor synthesis problems with partial observation}.
\bjtitle{Mathematics of Control, Signals and Systems}
\bvolume{2}(\bissue{1}),
\bfpage{47}--\blpage{69}
(\byear{1989})
\end{barticle}
\endbibitem

\bibitem[\protect\citeauthoryear{Checkoway et~al.}{2011}]{Checkoway:2011}
\begin{bchapter}
\bauthor{\bsnm{Checkoway}, \binits{S.}},
\bauthor{\bsnm{McCoy}, \binits{D.}},
\bauthor{\bsnm{Kantor}, \binits{B.}},
\bauthor{\bsnm{Anderson}, \binits{D.}},
\bauthor{\bsnm{Shacham}, \binits{H.}},
\bauthor{\bsnm{Savage}, \binits{S.}},
\bauthor{\bsnm{Koscher}, \binits{K.}},
\bauthor{\bsnm{Czeskis}, \binits{A.}},
\bauthor{\bsnm{Roesner}, \binits{F.}},
\bauthor{\bsnm{Kohno}, \binits{T.}}:
\bctitle{Comprehensive experimental analyses of automotive attack surfaces}.
In: \bbtitle{Proceedings of the 20th USENIX Conference on Security}.
\bsertitle{SEC'11},
pp. \bfpage{6}--\blpage{6}.
\bpublisher{USENIX Association},
\blocation{Berkeley, CA, USA}
(\byear{2011})
\end{bchapter}
\endbibitem

\bibitem[\protect\citeauthoryear{Carvalho et~al.}{2018}]{Carvalho:2018}
\begin{barticle}
\bauthor{\bsnm{Carvalho}, \binits{L.K.}},
\bauthor{\bsnm{Wu}, \binits{Y.C.}},
\bauthor{\bsnm{Kwong}, \binits{R.}},
\bauthor{\bsnm{Lafortune}, \binits{S.}}:
\batitle{Detection and mitigation of classes of attacks in supervisory control systems}.
\bjtitle{Automatica}
\bvolume{97},
\bfpage{121}--\blpage{133}
(\byear{2018})
\end{barticle}
\endbibitem

\bibitem[\protect\citeauthoryear{Easterly}{2023}]{Easterly:2023}
\begin{botherref}
\oauthor{\bsnm{Easterly}, \binits{J.}}:
The Attack on Colonial Pipeline: What We’ve Learned \& What We’ve Done Over the Past Two Years.
\url{https://www.cisa.gov/news-events/news/attack-colonial-pipeline-what-weve-learned-what-weve-done-over-past-two-years}.
Accessed: 2023-09-22
(2023)
\end{botherref}
\endbibitem

\bibitem[\protect\citeauthoryear{Fahim and Meira-Góes}{2024}]{Fahim2024-wodes}
\begin{barticle}
\bauthor{\bsnm{Fahim}, \binits{P.}},
\bauthor{\bsnm{Meira-Góes}, \binits{R.}}:
\batitle{Detecting probability footprints of sensor deception attacks in supervisory control}.
\bjtitle{17th IFAC Workshop on discrete Event Systems WODES 2024}
\bvolume{58}(\bissue{1}),
\bfpage{192}--\blpage{197}
(\byear{2024})
\end{barticle}
\endbibitem

\bibitem[\protect\citeauthoryear{Farwell and Rohozinski}{2011}]{Farwell:2011}
\begin{barticle}
\bauthor{\bsnm{Farwell}, \binits{J.P.}},
\bauthor{\bsnm{Rohozinski}, \binits{R.}}:
\batitle{Stuxnet and the future of cyber war}.
\bjtitle{Survival}
\bvolume{53}(\bissue{1}),
\bfpage{23}--\blpage{40}
(\byear{2011})
\end{barticle}
\endbibitem

\bibitem[\protect\citeauthoryear{Fritz and Zhang}{2023}]{fritz2023detection}
\begin{barticle}
\bauthor{\bsnm{Fritz}, \binits{R.}},
\bauthor{\bsnm{Zhang}, \binits{P.}}:
\batitle{Detection and localization of stealthy cyber attacks in cyber-physical discrete event systems}.
\bjtitle{IEEE Transactions on Automatic Control}
\bvolume{68}(\bissue{12}),
\bfpage{7895}--\blpage{7902}
(\byear{2023})
\end{barticle}
\endbibitem

\bibitem[\protect\citeauthoryear{Garg et~al.}{1999}]{Garg:1999}
\begin{barticle}
\bauthor{\bsnm{Garg}, \binits{V.K.}},
\bauthor{\bsnm{Kumar}, \binits{R.}},
\bauthor{\bsnm{Marcus}, \binits{S.I.}}:
\batitle{A probabilistic language formalism for stochastic discrete-event systems}.
\bjtitle{IEEE Transactions on Automatic Control}
\bvolume{44}(\bissue{2}),
\bfpage{280}--\blpage{293}
(\byear{1999})
\end{barticle}
\endbibitem

\bibitem[\protect\citeauthoryear{Greenberg}{2020}]{Greenberg:2020}
\begin{botherref}
\oauthor{\bsnm{Greenberg}, \binits{A.}}:
Dutch Hackers Found a Simple Way to Mess With Traffic Lights.
\url{https://www.wired.com/story/hacking-traffic-lights-netherlands/}.
Accessed: 2023-02-20
(2020)
\end{botherref}
\endbibitem

\bibitem[\protect\citeauthoryear{Hadjicostis et~al.}{2022}]{hadjicostis2022cybersecurity}
\begin{bchapter}
\bauthor{\bsnm{Hadjicostis}, \binits{C.N.}},
\bauthor{\bsnm{Lafortune}, \binits{S.}},
\bauthor{\bsnm{Lin}, \binits{F.}},
\bauthor{\bsnm{Su}, \binits{R.}}:
\bctitle{Cybersecurity and supervisory control: A tutorial on robust state estimation, attack synthesis, and resilient control}.
In: \bbtitle{2022 IEEE 61st Conference on Decision and Control (CDC)},
pp. \bfpage{3020}--\blpage{3040}
(\byear{2022}).
\bcomment{IEEE}
\end{bchapter}
\endbibitem

\bibitem[\protect\citeauthoryear{Kumar and Garg}{2001}]{Kumar:2001}
\begin{barticle}
\bauthor{\bsnm{Kumar}, \binits{R.}},
\bauthor{\bsnm{Garg}, \binits{V.K.}}:
\batitle{Control of stochastic discrete event systems modeled by probabilistic languages}.
\bjtitle{IEEE Transactions on Automatic Control}
\bvolume{46}(\bissue{4}),
\bfpage{593}--\blpage{606}
(\byear{2001})
\end{barticle}
\endbibitem

\bibitem[\protect\citeauthoryear{Kang et~al.}{2025}]{kang2025diagnosability}
\begin{barticle}
\bauthor{\bsnm{Kang}, \binits{T.}},
\bauthor{\bsnm{Seatzu}, \binits{C.}},
\bauthor{\bsnm{Li}, \binits{Z.}},
\bauthor{\bsnm{Giua}, \binits{A.}}:
\batitle{A joint diagnoser approach for diagnosability of discrete event systems under attack}.
\bjtitle{Automatica}
\bvolume{172},
\bfpage{112004}
(\byear{2025})
\end{barticle}
\endbibitem

\bibitem[\protect\citeauthoryear{Lima et~al.}{2019}]{Lima:2019}
\begin{barticle}
\bauthor{\bsnm{Lima}, \binits{P.M.}},
\bauthor{\bsnm{Alves}, \binits{M.V.S.}},
\bauthor{\bsnm{Carvalho}, \binits{L.K.}},
\bauthor{\bsnm{Moreira}, \binits{M.V.}}:
\batitle{Security against communication network attacks of cyber-physical systems}.
\bjtitle{Journal of Control, Automation and Electrical Systems}
\bvolume{30}(\bissue{1}),
\bfpage{125}--\blpage{135}
(\byear{2019})
\end{barticle}
\endbibitem

\bibitem[\protect\citeauthoryear{Li et~al.}{2025}]{li2025diagnosability}
\begin{barticle}
\bauthor{\bsnm{Li}, \binits{Y.}},
\bauthor{\bsnm{Hadjicostis}, \binits{C.N.}},
\bauthor{\bsnm{Wu}, \binits{N.}},
\bauthor{\bsnm{Li}, \binits{Z.}}:
\batitle{Tamper-tolerant diagnosability analysis and tampering detectability in discrete event systems under cost constraints}.
\bjtitle{Automatica}
\bvolume{171},
\bfpage{111971}
(\byear{2025})
\end{barticle}
\endbibitem

\bibitem[\protect\citeauthoryear{Lin et~al.}{2024}]{lin2024diagnosability}
\begin{barticle}
\bauthor{\bsnm{Lin}, \binits{F.}},
\bauthor{\bsnm{Lafortune}, \binits{S.}},
\bauthor{\bsnm{Wang}, \binits{C.}}:
\batitle{Diagnosability and attack detection for discrete event systems under sensor attacks}.
\bjtitle{Discrete Event Dynamic Systems}
\bvolume{34}(\bissue{3}),
\bfpage{465}--\blpage{495}
(\byear{2024})
\end{barticle}
\endbibitem

\bibitem[\protect\citeauthoryear{Lawford and Wonham}{1993}]{Lawford:1993}
\begin{bchapter}
\bauthor{\bsnm{Lawford}, \binits{M.}},
\bauthor{\bsnm{Wonham}, \binits{W.M.}}:
\bctitle{Supervisory control of probabilistic discrete event systems}.
In: \bbtitle{Proceedings of 36th Midwest Symposium on Circuits and Systems},
pp. \bfpage{327}--\blpage{331}
(\byear{1993})
\end{bchapter}
\endbibitem

\bibitem[\protect\citeauthoryear{Meira-G{\'o}es et~al.}{2020}]{meira-goes:2020synthesis}
\begin{barticle}
\bauthor{\bsnm{Meira-G{\'o}es}, \binits{R.}},
\bauthor{\bsnm{Kang}, \binits{E.}},
\bauthor{\bsnm{Kwong}, \binits{R.H.}},
\bauthor{\bsnm{Lafortune}, \binits{S.}}:
\batitle{Synthesis of sensor deception attacks at the supervisory layer of cyber--physical systems}.
\bjtitle{Automatica}
\bvolume{121},
\bfpage{109172}
(\byear{2020})
\end{barticle}
\endbibitem

\bibitem[\protect\citeauthoryear{Meira-G{\'o}es et~al.}{2020}]{meira-goes:2020towards}
\begin{barticle}
\bauthor{\bsnm{Meira-G{\'o}es}, \binits{R.}},
\bauthor{\bsnm{Keroglou}, \binits{C.}},
\bauthor{\bsnm{Lafortune}, \binits{S.}}:
\batitle{Towards probabilistic intrusion detection in supervisory control of discrete event systems}.
\bjtitle{21st IFAC World Congress}
\bvolume{53}(\bissue{2}),
\bfpage{1776}--\blpage{1782}
(\byear{2020})
\end{barticle}
\endbibitem

\bibitem[\protect\citeauthoryear{Meira-G{\'o}es et~al.}{2021}]{meira-goes:2021synthesis}
\begin{barticle}
\bauthor{\bsnm{Meira-G{\'o}es}, \binits{R.}},
\bauthor{\bsnm{Kwong}, \binits{R.H.}},
\bauthor{\bsnm{Lafortune}, \binits{S.}}:
\batitle{Synthesis of optimal multiobjective attack strategies for controlled systems modeled by probabilistic automata}.
\bjtitle{IEEE Transactions on Automatic Control}
\bvolume{67}(\bissue{6}),
\bfpage{2873}--\blpage{2888}
(\byear{2021})
\end{barticle}
\endbibitem

\bibitem[\protect\citeauthoryear{Meira-G{\'o}es et~al.}{2021}]{meira-goes2021synthesistac}
\begin{barticle}
\bauthor{\bsnm{Meira-G{\'o}es}, \binits{R.}},
\bauthor{\bsnm{Lafortune}, \binits{S.}},
\bauthor{\bsnm{Marchand}, \binits{H.}}:
\batitle{Synthesis of supervisors robust against sensor deception attacks}.
\bjtitle{IEEE Transactions on Automatic Control}
\bvolume{66}(\bissue{10}),
\bfpage{4990}--\blpage{4997}
(\byear{2021})
\end{barticle}
\endbibitem

\bibitem[\protect\citeauthoryear{Meira-G{\'o}es et~al.}{2023}]{meira-goes:2023dealing}
\begin{barticle}
\bauthor{\bsnm{Meira-G{\'o}es}, \binits{R.}},
\bauthor{\bsnm{Marchand}, \binits{H.}},
\bauthor{\bsnm{Lafortune}, \binits{S.}}:
\batitle{Dealing with sensor and actuator deception attacks in supervisory control}.
\bjtitle{Automatica}
\bvolume{147},
\bfpage{110736}
(\byear{2023})
\end{barticle}
\endbibitem

\bibitem[\protect\citeauthoryear{Oliveira et~al.}{2023}]{OLIVEIRA2023100907}
\begin{barticle}
\bauthor{\bsnm{Oliveira}, \binits{S.}},
\bauthor{\bsnm{Leal}, \binits{A.B.}},
\bauthor{\bsnm{Teixeira}, \binits{M.}},
\bauthor{\bsnm{Lopes}, \binits{Y.K.}}:
\batitle{A classification of cybersecurity strategies in the context of discrete event systems}.
\bjtitle{Annual Reviews in Control}
\bvolume{56},
\bfpage{100907}
(\byear{2023})
\end{barticle}
\endbibitem

\bibitem[\protect\citeauthoryear{Pantelic et~al.}{2014}]{Pantelic:2014}
\begin{barticle}
\bauthor{\bsnm{Pantelic}, \binits{V.}},
\bauthor{\bsnm{Lawford}, \binits{M.}},
\bauthor{\bsnm{Postma}, \binits{S.}}:
\batitle{A framework for supervisory control of probabilistic discrete event systems}.
\bjtitle{12th IFAC International Workshop on Discrete Event Systems (WODES)}
\bvolume{47}(\bissue{2}),
\bfpage{477}--\blpage{484}
(\byear{2014})
\end{barticle}
\endbibitem

\bibitem[\protect\citeauthoryear{Poor}{2013}]{poor2013introduction}
\begin{bbook}
\bauthor{\bsnm{Poor}, \binits{H.V.}}:
\bbtitle{An Introduction to Signal Detection and Estimation}.
\bpublisher{Springer},
\blocation{Heidelberg, Germany}
(\byear{2013})
\end{bbook}
\endbibitem

\bibitem[\protect\citeauthoryear{Ramadge and Wonham}{1987}]{Ramadge:1987}
\begin{barticle}
\bauthor{\bsnm{Ramadge}, \binits{P.J.}},
\bauthor{\bsnm{Wonham}, \binits{W.M.}}:
\batitle{Supervisory control of a class of discrete event processes}.
\bjtitle{SIAM J. Control Optim.}
\bvolume{25}(\bissue{1}),
\bfpage{206}--\blpage{230}
(\byear{1987})
\end{barticle}
\endbibitem

\bibitem[\protect\citeauthoryear{{Rashidinejad} et~al.}{2019}]{Rashidinejad:2019}
\begin{bchapter}
\bauthor{\bsnm{{Rashidinejad}}, \binits{A.}},
\bauthor{\bsnm{{Wetzels}}, \binits{B.}},
\bauthor{\bsnm{{Reniers}}, \binits{M.}},
\bauthor{\bsnm{{Lin}}, \binits{L.}},
\bauthor{\bsnm{{Zhu}}, \binits{Y.}},
\bauthor{\bsnm{{Su}}, \binits{R.}}:
\bctitle{Supervisory control of discrete-event systems under attacks: An overview and outlook}.
In: \bbtitle{2019 18th European Control Conference (ECC)},
pp. \bfpage{1732}--\blpage{1739}
(\byear{2019})
\end{bchapter}
\endbibitem

\bibitem[\protect\citeauthoryear{Su}{2018}]{Su:2018}
\begin{barticle}
\bauthor{\bsnm{Su}, \binits{R.}}:
\batitle{Supervisor synthesis to thwart cyber attack with bounded sensor reading alterations}.
\bjtitle{Automatica}
\bvolume{94},
\bfpage{35}--\blpage{44}
(\byear{2018})
\end{barticle}
\endbibitem

\bibitem[\protect\citeauthoryear{Thorsley and Teneketzis}{2006}]{Thorsley:2006}
\begin{bchapter}
\bauthor{\bsnm{Thorsley}, \binits{D.}},
\bauthor{\bsnm{Teneketzis}, \binits{D.}}:
\bctitle{Intrusion detection in controlled discrete event systems}.
In: \bbtitle{Proceedings of the 45th IEEE Conference on Decision and Control},
pp. \bfpage{6047}--\blpage{6054}
(\byear{2006})
\end{bchapter}
\endbibitem

\bibitem[\protect\citeauthoryear{Tong et~al.}{2022}]{tong:2022}
\begin{barticle}
\bauthor{\bsnm{Tong}, \binits{Y.}},
\bauthor{\bsnm{Wang}, \binits{Y.}},
\bauthor{\bsnm{Giua}, \binits{A.}}:
\batitle{A polynomial approach to verifying the existence of a threatening sensor attacker}.
\bjtitle{IEEE Control Systems Letters}
\bvolume{6},
\bfpage{2930}--\blpage{2935}
(\byear{2022})
\end{barticle}
\endbibitem

\bibitem[\protect\citeauthoryear{Wonham and Cai}{2018}]{Wonham:2018}
\begin{bbook}
\bauthor{\bsnm{Wonham}, \binits{W.M.}},
\bauthor{\bsnm{Cai}, \binits{K.}}:
\bbtitle{Supervisory Control of Discrete-Event Systems}.
\bpublisher{Springer},
\blocation{Heidelberg, Germany}
(\byear{2018})
\end{bbook}
\endbibitem

\bibitem[\protect\citeauthoryear{Wang and Tong}{2022}]{wang:2022}
\begin{bchapter}
\bauthor{\bsnm{Wang}, \binits{K.}},
\bauthor{\bsnm{Tong}, \binits{Y.}}:
\bctitle{Sensor and actuator attack identification in discrete event systems}.
In: \bbtitle{2022 41st Chinese Control Conference (CCC)},
pp. \bfpage{1605}--\blpage{1610}
(\byear{2022})
\end{bchapter}
\endbibitem

\bibitem[\protect\citeauthoryear{Yoo and Lafortune}{2002}]{yoo2002polynomial}
\begin{barticle}
\bauthor{\bsnm{Yoo}, \binits{T.-S.}},
\bauthor{\bsnm{Lafortune}, \binits{S.}}:
\batitle{Polynomial-time verification of diagnosability of partially observed discrete-event systems}.
\bjtitle{IEEE Transactions on automatic control}
\bvolume{47}(\bissue{9}),
\bfpage{1491}--\blpage{1495}
(\byear{2002})
\end{barticle}
\endbibitem

\bibitem[\protect\citeauthoryear{Yao et~al.}{2024}]{YAO2024deception}
\begin{barticle}
\bauthor{\bsnm{Yao}, \binits{J.}},
\bauthor{\bsnm{Li}, \binits{S.}},
\bauthor{\bsnm{Yin}, \binits{X.}}:
\batitle{Sensor deception attacks against security in supervisory control systems}.
\bjtitle{Automatica}
\bvolume{159},
\bfpage{111330}
(\byear{2024})
\end{barticle}
\endbibitem

\bibitem[\protect\citeauthoryear{Zhang}{2023}]{zhang2023robust}
\begin{barticle}
\bauthor{\bsnm{Zhang}, \binits{Q.}}:
\batitle{Robust predictability in discrete event systems under sensor attacks}.
\bjtitle{Frontiers in Physics}
\bvolume{11},
\bfpage{299}
(\byear{2023})
\end{barticle}
\endbibitem

\end{thebibliography}

\appendix
\section{Appendix}
\subsection{Probabilistic Parallel Composition $||_p$} \label{app:parallel_prob}
\begin{definition}
The probabilistic composition, $||_p$ between PDES $G$ and DFA $R$ is defined by $R||_p G = (X_R\times X_G, \Sigma, \delta_{R,G},P_{R,G}, (x_{0,R},x_{0,G}), X_{R}\times X_{m,G})$ where:
\begin{equation}
P_{R,G}((x_R,x_G),e,(y_R,y_G)) = \left\lbrace
\begin{array}{ll}
\frac{P_G(x_G,e)}{\sum_{\sigma\in\Gamma_G(x_R)\cap\Gamma_R(y_R)}P_G(x_G,\sigma)} &
\ e\in \Gamma_R(x_R)\cap\Gamma_G(x_G) \\
0 & \text{otherwise}
\end{array}
\right.
\end{equation}
\end{definition}

\subsection{Construction of attack strategy using automaton $A$}\label{app:A_aut}

An automaton $A= (X_A, \Sigma_m, \delta_{A},x_{0,A})$ encodes an attack strategy if the following condition holds:
\begin{enumerate}
    \item[(1)] $\forall x\in X_A$, $e\in \Sigma_i$. $\delta_A(x,e)! \Rightarrow \forall e\in \Sigma_m\setminus\{e\}.\  \delta_A(x,e)\hspace{-0.1cm}\not{!}$ 
    \item[(2)] $\forall x\in X_A$, $e\in \Sigma_a$. $\delta_A(x,e)! \Leftrightarrow \delta_A(x,del(e))\hspace{-0.1cm}\not{!}$ 
\end{enumerate}
Condition (1) ensures that states with an insertion event do not have any other transitions. 
Since the attacker is deterministic, it can only insert one event when it has ``decided" to insert.
Condition (2) ensures that the attacker can only delete a compromised event or not attack this event, i.e., it cannot have both options since it is deterministic.

Next, we show how to extract an attack strategy from an automaton $A$.
First, for state $q\in X_A$, the function $Ins(q)$ returns a string of insertion events that can be executed from $q$ until it reaches a state in $A$ where it cannot insert events.
The function $Ins(q)$ is defined recursively as follows:
\begin{equation}
Ins(q) = \left\{ 
\begin{array}{cc}
  eIns(\delta_A(q,e))   & \text{if } \Gamma_A(q) = \{e\} \wedge e\in \Sigma_i\\
  \epsilon  &  \text{if } \Gamma_A(q)\cap \Sigma_i = \empty
\end{array}
\right.
\end{equation}

Now, we are ready to extract a strategy from $A$.
\begin{definition}
Given DFA $A = (X_A, \Sigma_m, \delta_{A},x_{0,A})$ satisfying the conditions above.
First, the strategy for $\epsilon$ is $A(\epsilon,\epsilon) = Ins(x_{0,A})$, i.e., a string of insertions from the initial state if they are present in A.
Next, the attack strategy generated by $A$ is constructed as follows for any state $q\in X_A$ and event $e\in \Sigma$:

\begin{equation}\label{eq:}
A(q,e) = \left\lbrace
\begin{array}{ll}
del(e) Ins(\delta_A(q,del(e)) & \text{if } del(e) \in \Gamma_A(q)\\
e Ins(\delta_A(q,e)) & \text{if } e\in \Gamma_A(q)\\
\text{undefined} & \text{otherwise}
\end{array}
\right.
\end{equation}
\end{definition}
The strategy is defined for state $q$ after observing event $e\in \Gamma_A(q)\cap \Sigma$ is: (1) maintain observation of event $e$ followed by a possible string of insertion event, or (2) replace $e$ with $del(e)$ followed by a possible string of insertion event.
Lastly, we can extend the function $A(q,e)$ for strings as in Def.~\ref{def:attack_str} for any  $s\in \lang(A)$ as $A(s,e) = A(\delta_A(x_{0,A},s),e)$.
The attack strategy is also undefined for any string $s\notin \lang(A)$.

\end{document}